\documentclass[conference]{IEEEtran}

\usepackage{graphicx}
\usepackage{epstopdf}
\usepackage{listings}
\usepackage{xcolor}
\usepackage{amsmath, amsfonts, amssymb}
\usepackage{graphicx,wrapfig,floatflt}

\usepackage{array}
\usepackage{mdwmath}
\usepackage{mdwtab}
\usepackage{eqparbox}
\usepackage{bm,bbm,amssymb,amsthm,url}

\usepackage{algorithm}
\usepackage[noend]{algpseudocode}

\makeatletter
\def\BState{\State\hskip-\ALG@thistlm}
\makeatother

\ifCLASSOPTIONcompsoc \usepackage[caption=false,font=normalsize,labelfont=sf,textfont=sf]{subfig}
\else
\usepackage[caption=false,font=footnotesize]{subfig}
\fi

\numberwithin{equation}{section}

\newcommand{\MAT}{\left[ \begin{array}}  
\newcommand{\mat}{\end{array} \right]}

\newtheorem{Lemma}{Lemma}[section]
\newtheorem{Theorem}{Theorem}[section]

\newtheorem{Q}{Question}[section]
\newtheorem{Proposition}{Proposition}

\def \minimize {\operatorname*{minimize}}

\def \tr {\operatorname*{Tr\ }}

\def \Re {\operatorname*{Re}}
\def \diag {\operatorname*{diag}}

\def \tPPhi {\mathbf{\Phi}}
\def \ts {\s}
\def \thX {\mathbf{\hat{\X}}}
\def \PPhi {\mathbf{\Phi}}
\def \pphi {\mathbf{\phi}}
\def \0 {\mathbf{0}}
\def \a {\bm{a}}

\def \A {\mathbf{A}}

\def \b {\bm{b}}
\def \B {\mathbf{B}}

\def \c {\bm{c}}
\def \C {\mathbf{C}}

\def \D {\mathbf{D}}

\def \e {\bm{e}}

\def \h {\bm{h}}

\def \H {\mathbf{H}}
\def \I {\mathbf{I}}

\def \L {\mathcal{L}}

\def \n {\bm{n}}

\def \P {\mathbf{P}}

\def \R {\mathbf{R}}

\def \s {\bm{s}}

\def \v {\bm{v}}

\def \x {\bm{x}}

\def \X {\mathbf{X}}

\def \y {\bm{y}}

\def \z {\bm{z}}

\begin{document}
%
\title{Support Recovery for Sparse Signals with Unknown Non-stationary Modulation}
%
%
%

\author{Youye~Xie,\IEEEmembership{~Student Member,~IEEE,}
        Michael~B.~Wakin,\IEEEmembership{~Senior Member,~IEEE,}
        and Gongguo~Tang\IEEEmembership{~Member,~IEEE}
        \\
        Department
of Electrical Engineering, Colorado School of Mines, USA
\thanks{The authors are with the Department
of Electrical Engineering, Colorado School of Mines, Golden, CO 80401, USA, e-mails: youyexie@mines.edu, mwakin@mines.edu@mines.edu, gtang@mines.edu.}}

%
%

\markboth{Journal of \LaTeX\ Class Files,~Vol.~14, No.~8, August~2015}%
{Shell \MakeLowercase{\textit{et al.}}: Bare Demo of IEEEtran.cls for IEEE Journals}
%



\maketitle

\begin{abstract}
The problem of estimating a sparse signal from low dimensional noisy observations arises in many applications, including super resolution, signal deconvolution, and radar imaging. In this paper, we consider a sparse signal model with non-stationary modulations, in which each dictionary atom contributing to the observations undergoes an unknown, distinct modulation. By applying the lifting technique, under the assumption that the modulating signals live in a common subspace, we recast this sparse recovery and non-stationary blind demodulation problem as the recovery of a column-wise sparse matrix from structured linear observations, and propose to solve it via block $\ell_{1}$-norm regularized quadratic minimization. Due to observation noise, the sparse signal and modulation process cannot be recovered exactly. Instead, we aim to recover the sparse support of the ground truth signal and bound the recovery errors of the signal's non-zero components and the modulation process. In particular, we derive sufficient conditions on the sample complexity and regularization parameter for exact support recovery and bound the recovery error on the support. Numerical simulations verify and support our theoretical findings, and we demonstrate the effectiveness of our model in the application of single molecule imaging.
%
\end{abstract}


%
\IEEEpeerreviewmaketitle

\section{Introduction}

\subsection{Overview}
\footnote{This paper was presented in part at the 53rd Asilomar Conference on Signals, Systems, and Computers, 2019 \cite{xie2019asilomar}.}The problem of recovering a high dimensional sparse signal from its low dimension observations using a fixed sensing mechanism arises naturally in a wide range of applications, including radar autofocus \cite{mansour2018radar}, magnetic resonance imaging \cite{lustig2007sparse}, and video acquisition \cite{mansour2012adaptive}. Typically, the system receives a low dimensional signal $\y=\D\A\c\in\C^{N}$, where $\c\in\C^{M}$ ($M>N$) is an unknown high dimensional signal, and $\D$ and $\A$ are  known sensing matrices. Although the sensing process is under-determined, one can solve for $\c$ by leveraging its sparsity; this sparse recovery problem has been studied extensively by the compressive sensing community \cite{candes2008introduction,candes2011probabilistic,xie2017radar}.

When $\D\in\C^{N\times N}$ is a diagonal matrix containing a sampled carrier signal along its diagonal, it describes a modulation process, and thus recovery with unknown diagonal $\D$ is sometimes referred to as simultaneous sparse recovery and blind demodulation \cite{xie2019simultaneous}. Scenarios where $\D$ is unknown arise in certain self-calibration \cite{ling2015self} and blind deconvolution problems \cite{ahmed2014blind}.

In this paper, we further generalize this model, allowing each atom in the dictionary matrix $\A$ to undergo a distinct modulation process, rather than multiplication by the same matrix $\D$. We refer to this generalized scenario as {\em non-stationary modulation}. Moreover, we suppose that the observation is contaminated with random noise. Although we no longer expect to recover the sparse vector $\c$ and modulating signals (which we denote as $\D_j$) exactly due to the existence of noise, we focus on recovering the sparse support of $\c$ and on bounding the recovery error of $\c$ and $\D_j$. By employing the lifting technique and under the assumption that the modulating signals live in a known, common subspace, we recast our problem as the recovery of a column-wise sparse matrix from structured linear observations. Under this formulation, there are no unknown parameters in the lifted linear operator. We solve the support recovery problem by solving a block $\ell_l$-norm ($\ell_{2,1}$-norm) regularized quadratic minimization problem, which is also known as the group lasso in the statistics literature \cite{zou2006adaptive,yuan2006model}. The generalized model encompasses a wide range of applications, including direction of arrival (DOA) estimation for an antenna array with DOA sensitive channel responses \cite{friedlander2014bilinear}, frequency estimation with damping in nuclear magnetic resonance spectroscopy \cite{yang2016super}, and CDMA communication with a spreading sequence sensitive channel \cite{ling2015self}. To give a concrete example, we apply the proposed model to single molecule imaging \cite{tang2013sparse} in Section \ref{SMI}.

\subsection{Setup and Notation\label{notation}}

Throughout the paper, we represent matrices, vectors, and scalars as bold uppercase, $\X$, bold lower case, $\x$, and non-bold letters, $x$, respectively. We use the symbol $C$ to denote numerical constants that might vary from line to line. Given a support set $T$, the notation $\X_T$ represents the restriction of $\X$ to the columns indexed by $T$, and the notation $\x_T$ represents the restriction of $\x$ to the entries indexed by $T$. Moreover, we use $||\cdot||$ to denote the spectral norm, which returns the maximum singular value of a matrix, and $||\cdot||_F$ to denote the Frobenius norm. For a matrix $\X=[\x_1, \x_2,\cdots,\x_M]\in\C^{K\times M}$, we define $||\X||_{2,1}=\sum_{j=1}^M||\x_j||_2$ and $||\X||_{2,\infty}=\max_{j}||\x_j||_2$. In addition, later in the paper we will have the vectorized subgradient, $\s\in \C^{KM\times 1}$, of a function with respect to its matrix input $\X\in\C^{K\times M}$, and we define $||\s||_{2,\infty} = \max_j||\s_j||_2$ where $\s_j$ is the subgradient with respect to $\x_j$.

\subsection{Problem Formulation}
In this paper, we consider the following generalized signal model with an unknown coefficient vector and non-stationary modulation process. Specifically, the observations consist of a contaminated composite signal
\begin{align}
\y=\sum^{M}_{j=1}c_j\D_j\a_j + \n\in\C^N.
\label{origin}
\end{align}
Here $c_j\in\C$ is an unknown scalar, $\D_j\in\C^{N\times N}$ is an unknown modulation matrix which is non-stationary as it depends on $j$, $\a_j$ is a dictionary atom coming from a dictionary matrix $\A=[\a_1, \a_2,\cdots,\a_M]\in\C^{N\times M}$, and $\n\in\C^{N\times 1}$ is additive random Gaussian noise whose real and imaginary entries follow the i.i.d Gaussian distribution with mean $0$ and variance $\sigma^2$.

Since there are more unknown parameters than the number of observations in the model \eqref{origin}, to make the recovery problem well-posed, we assume that at most $J$ $(<M)$ of the coefficients $c_j$ are non-zero and that the diagonal modulation matrices, $\D_j$, live in a common $K$-dimension subspace
\begin{align}
\mathbf{D}_j = \diag(\mathbf{B}\h_j)
\label{eq:sc}
\end{align}
where $\B\in\C^{N\times K}$ $(N>K)$ is a known basis for the subspace with orthonormal columns, and $\h_j\in\C^{K\times 1}$ are unknown coefficient vectors. Similar subspace assumptions can be found in the deconvolution and demixing literature \cite{ling2017blind, chi2016guaranteed}. Recovering $c_j$ and $\h_j$ from $\y$ is a bilinear inverse problem \cite{ling2018self,aghasi2018convex}.

To combat the difficulties resulting from the bilinearity, we apply the lifting trick \cite{xie2019simultaneous, ling2017blind,asif2009random}, which collects the unknown parameters into a matrix $\X=[c_1\h_1\quad c_2\h_2\quad\cdots\quad c_M\h_M]\in\C^{K\times M}$. By using Proposition~1 in~\cite{xie2019SPARSE} we can show that, when $\n=\mathbf{0}$, the observation model \eqref{origin} takes the following equivalent form:
\begin{align}
\y(n)=\b'^H_n\X\a'_n, n = 1, \ldots, N.
\label{L}
\end{align}
where $\b'_n$ and $\a_n'$ are the $n$-th column of $\B^H$ and $\A^T$ respectively. We write \eqref{L} succinctly as $\y=\L(\X)$ with $\L$ being a properly defined linear operator. And the adjoint of the linear operator $\L$ is $\L^*(\y)=\sum_{l=1}^Ny_l\b'_l\a'^H_{l}$. The matrix $\X$ incorporates the unknown sparse signal and modulation process with at most $J$~$(<M)$ non-zero columns. The support recovery problem we study in this paper aims to determine the indices, $j$, of the non-zero columns in $\X$ from the observation vector $\y$. We also aim to bound the recovery error of $\X$ in terms of the $\ell_{2,\infty}$-norm. If we assume there is no trivial null modulation, namely all $\D_j\neq0$, finding the indices of the non-zero columns of $\X$ is equivalent to recovering the support of $\c$. Moreover, note that due to the scaling ambiguity between $c_j$ and $\h_j$, the recovery error bound is expressed with respect to their multiplication $c_j\h_j$.

A natural way to recover the ground truth $\X_0$ from $\y$ is to exploit its sparse property and solve the following $\ell_{2,1}$-norm regularized quadratic minimization problem
\begin{equation}
\begin{aligned}
&\minimize_{\X\in\C^{K\times M}} \frac{1}{2}||\y-\L(\X)||_2^2+\lambda||\X||_{2,1}.
\label{21}
\end{aligned}
\end{equation}
Alternatively, we can write
\eqref{21} equivalently as
\begin{equation}
\begin{aligned}
&\minimize_{\X\in\C^{K\times M}} \frac{1}{2}||\y-\PPhi\cdot\text{vec}(\X)||_2^2+\lambda\sum_{i=1}^M||\x_i||_{2}.
\label{glasso}
\end{aligned}
\end{equation}
Here $\L(\X)=\PPhi\cdot\text{vec}(\X)$ with
\begin{equation}
\begin{aligned}
\PPhi=[ \pphi_{1,1}\quad\cdots\quad\pphi_{K,1}\quad\cdots\quad\pphi_{1,M}\quad\cdots\quad&\pphi_{K,M} ]\\
&\in\C^{N\times KM}
\label{phi}
\end{aligned}
\end{equation}
and $\pphi_{i,j}=\text{diag}(\b_i)\a_j\in\C^{N\times 1}$, where $\b_i$ is the $i$-th column of $\B$. Moreover, we denote the set containing the indices of the non-zero columns of the ground-truth matrix $\X_0$ as $T := T(\X_0)$ with $|T|=J$ and its complement as $T^C$. Due to the special block structure of $\PPhi$, when using the subscript notation $\PPhi_T$ we refer to the $N\times KJ$ sub-matrix of $\PPhi$ containing the $K(j-1)+1$ to $K(j-1)+K$-th columns for all $j\in T$.

\subsection{Main Contributions}

Our contributions are twofold. First, we propose to apply $\ell_{2,1}$-norm regularized quadratic minimization to recover the support of the generalized signal model in \eqref{origin}. Second, we derive sufficient conditions under which, with overwhelming probability, the support of the recovered signal is a subset of the support of the ground truth. More precisely, we show that the required number of observations, $N$, is proportional to the number of degrees of freedom, $O(JK)$, up to logarithmic factors. Moreover, the regularization parameter, $\lambda$, should be chosen to be proportional to the $\sigma$ of the noise. We also bound the error in recovering the non-zero columns of the ground truth as measured in the $\ell_{2,\infty}$-norm. With an additional assumption on the ground truth signal, all conditions lead to exact support recovery.

\subsection{Related Work}

The $\ell_{2,1}$-norm constrained quadratic minimization problem, also known as the group lasso in statistics literature \cite{zou2006adaptive,yuan2006model,nardi2008asymptotic}, has been widely studied. However, under our particular signal model \eqref{origin}, the linear operator $\PPhi$ contains randomness and has a special block structure as presented in \eqref{phi}, which distinguishes our work from other group lasso research. For example, \cite{yuan2006model} assumes each block of $\PPhi$, $[\phi_{1,j},...,\phi_{K,j}]$, to be orthonormal. \cite{lv2011group} considers the adaptive group lasso and derives sufficient support recovery conditions using the block coherence of a deterministic $\PPhi$. \cite{huang2010benefit} allows varying block sizes but still assumes a deterministic $\PPhi$. \cite{sivakumar2015beyond} assumes that $\PPhi$ has independent sub-exponential rows which is not consistent with our formulation, and they bound the recovery error in terms of $\ell_2$-norm instead of $\ell_{2,\infty}$-norm as in our theorem. Moreover, \cite{vaiter2015model,vaiter2017model} provide a general recovery analysis for regression problems regularized with partly smooth functions relative to a manifold defined in \cite{vaiter2015model}, which encompasses the $\ell_{2,1}$-norm. However, the precise bounds on the regularization parameter and sample complexity for exact support recovery with $\PPhi$ defined in \eqref{phi} are not derived, and that work bounds the error in terms of $\ell_2$-norm instead of the $\ell_{2,\infty}$-norm.

As for the signal model itself, the model we study is closely related to certain works in self-calibration and blind deconvolution \cite{ling2015self,ahmed2014blind}. The work in \cite{chi2016guaranteed} considers a similar model except that the dictionary therein consists of all sampled sinusoids over a continuous frequency range, and its modulating waveforms, $\D_j$, are all the same. As an extension, \cite{yang2016super} allows non-stationary modulating waveforms but still concerns the sinusoid dictionary. The fact that \cite{yang2016super} considers a more general signal model than \cite{chi2016guaranteed} actually facilitates the derivation of a near optimal result on the sufficient sample complexity. Our work in this paper similarly benefits from expanding the signal model of \cite{ling2015self}. Specifically, our model fits into the self-calibration problem \cite{ling2015self} when all $\D_j$ are the same. However, in the noisy case, \cite{ling2015self} does not aim to recover the support and only bounds the error in terms of the $\ell_2$-norm. \cite{ling2017blind} generalizes the model in \cite{ling2015self} and can be interpreted as the self-calibration with multiple sensors, while allowing varying calibration parameters. However, \cite{ling2017blind}  studies a constrained nuclear norm minimization problem with bounded noise and requires knowing the number of sensors. Additional related models for different applications, all requiring the same modulation matrix, are available in \cite{ahmed2014blind,hung2017low,eldar2017sensor,flinth2018sparse}.

We have also previously studied the sparse recovery and blind demodulation problem \cite{xie2019simultaneous,xie2019SPARSE} and numerically compared the support recovery performance of the SparseLift method \cite{ling2015self} and the $\ell_{2,1}$-norm minimization method for direction of arrival estimation in \cite{xie2019simultaneous}. In those works, however, we assume either zero or bounded additive noise, whereas we consider random Gaussian noise in this paper. Moreover, in \cite{xie2019simultaneous,xie2019SPARSE} we solve a constrained $\ell_{2,1}$-norm minimization problem due to the consideration of bounded noise. The regularized formulation used in this paper is a natural choice when considering unbounded noise \cite{hastie2015statistical} and is more convenient for support recovery analysis. Finally, in those papers, we derive the recovery error bound in terms of the $\ell_2$-norm and do not study the question of exact support recovery when noise is involved.

\vspace{2mm}
The rest of the paper is organized as follows. In Section \ref{mainresult}, we present our main theorem regarding the support recovery problem. The detailed proof of the main theorem is shown in Section \ref{proofofmaintheorem}. Several simulations and an experiment are conducted in Section \ref{simulations} to demonstrate the important scaling relationships and the effectiveness of our model in practical application. Finally, we conclude this paper in Section \ref{conclusion}.

\section{Main Result\label{mainresult}}

In this section, we present our main theorem, which presents the support recovery conditions and recovery error bound for solving \eqref{21} (or equivalently \eqref{glasso}). In this result, we assume that the dictionary matrix $\A$ is a random Gaussian matrix, by which we mean a matrix whose entries follow the i.i.d standard normal distribution.

\begin{Theorem}
Consider the observation model in equation (\ref{origin}), assume that $\A\in\R^{N\times M}$ $(N<M)$ is a random Gaussian matrix, at most $J$ $(< M)$ coefficients $c_j$ are nonzero, and the real and imaginary parts of each entry of the noise vector $\n\in\C^{N\times 1}$ follow the i.i.d Gaussian distribution with 0 mean and $\sigma^2$ variance. Suppose also that each modulation matrix $\D_j$ satisfies the subspace constraint~\eqref{eq:sc}, where $\B^H\B=\I_K$. If the number of observations
\begin{align}
N\geq C_{\alpha,1}\mu_{max}^2JK\left[\log(M-J)+\log^2(N)\right]
\label{N}
\end{align}
and the regularization parameter
\begin{align}
\lambda\geq\sqrt{C_{\alpha,2}\sigma^2\mu_{max}^2K\left[\log(M-J)+\log(N)\right]}
\label{lambdabelow}
\end{align}
where $C_{\alpha,1}$ and $C_{\alpha,2}$ are constants that grow linearly with $\alpha>1$ and the coherence parameter
\begin{align*}
\mu_{max} = \max_{i,j}\sqrt{N}|\B_{ij}|,
\end{align*}
then the following properties hold with probability at least $1-O(N^{-\alpha+1})$:
\begin{enumerate}
\item Problem \eqref{glasso} has a unique solution $\hat{\X}\in\C^{K\times M}$ with its support, the set of indices of the non-zero columns in $\hat{\X}$, contained within the support $T$ of the ground truth solution, $\X_0$.

\item The recovery error between the solution, $\hat{\X}$, and the ground truth, $\X_0$, satisfies

\begin{equation}
\begin{aligned}
||\hat{\X}-\X_0||_{2,\infty}\leq \sqrt{C_{\alpha}\sigma^2\mu_{max}^2JK\left[\log(J)+\log(N)\right]}\\
+4\sqrt{J}\lambda
\end{aligned}
\label{errorbound}
\end{equation}
where $C_{\alpha}$ is a constant that grows linearly with $\alpha$. If in addition the non-zero columns of $\X_0$ are bounded below
\begin{equation}
\begin{aligned}
\min_{j\in T}||\x_{0,j}||_{2}>\sqrt{C_{\alpha}\sigma^2\mu_{max}^2JK\left[\log(J)+\log(N)\right]}\\
+4\sqrt{J}\lambda,
\end{aligned}
\label{x0below}
\end{equation}
then $\hat{\X}$ and $\X_0$ have exactly the same support which implies exact support recovery.
\end{enumerate}
\label{maintheorem}
\end{Theorem}

According to \eqref{errorbound}, we can derive that for any $\hat{\x}_j=\hat{c}_j\hat{\h}_j$ and $\x_{0,j}=c_{0,j}\h_{0,j}$ which are the $j$-th columns of the solution $\hat{\X}$ and the ground truth $\X_0$ respectively, $||\hat{c}_j\hat{\D}_j-c_{0,j}\D_{0,j}||_F=||\hat{c}_j\hat{\h}_j-c_{0,j}\h_{0,j}||_2\leq \sqrt{C_{\alpha}\sigma^2\mu_{max}^2JK\left[\log(J)+\log(N)\right]}+4\sqrt{J}\lambda$. Moreover, since the columns of $\B$ are orthonormal, $\mu_{max}\in[1,\sqrt{N}]$. Given the system parameters and a large enough $N$, \eqref{N} is satisfied when $1\leq\mu_{max}\leq\sqrt{\frac{N}{C_{\alpha,1}KJ\left[\log(M-J)+\log^2(N)\right]}}$. In addition, since we solve the column-wise sparse matrix support recovery problem via the group lasso and bound the recovery error in terms of $\ell_{2,\infty}$-norm, Theorem \ref{maintheorem} may be of interest outside the support recovery problem and shed light on the performance of the group lasso with random block structured linear operators.

\section{Proof of Theorem \ref{maintheorem}\label{proofofmaintheorem}}
We present proof of the main theorem in this section. We first derive the optimality and uniqueness conditions of the solution to \eqref{glasso} and then apply the primal-dual witness method \cite{wainwright2009sharp} to construct a solution and find the conditions regarding the regularization parameter $\lambda$ and number of observations $N$ such that the optimality and uniqueness conditions are satisfied.

\subsection{Optimality and Uniqueness Conditions}

\begin{Lemma}
\leavevmode
\begin{enumerate}
\item A matrix $\hat{\X}\in\C^{K\times M}$ is an optimal solution to \eqref{glasso} if and only if there exists a subgradient vector $\s=\begin{bmatrix}
\s_1\\
\vdots\\
\s_M
\end{bmatrix}\in\text{vec}\left(\partial||\hat{\X}||_{2,1}\right)$, such that
\begin{align*}
\PPhi^H\PPhi\cdot\text{vec}(\hat{\X})-\PPhi^H\y+\lambda\cdot
\begin{bmatrix}
\s_1\\
\vdots\\
\s_M
\end{bmatrix}=\0
\end{align*}
which is equivalent to
\begin{align}
\PPhi^H\PPhi\cdot\left(\text{vec}(\hat{\X})-\text{vec}(\X_0)\right)-\PPhi^H\n+\lambda\s=\0
\label{firstordercondition}
\end{align}
where $\s_i\in\C^{K}$ is the subgradient of $||\cdot||_2$ at $\hat{\x}_i$ defined as
\begin{align}
\s_i = \left\{ \begin{array}{ll}
         \frac{\hat{\x}_i}{||\hat{\x}_i||_2} & \mbox{if $||\hat{\x}_i||_2\neq 0$};\\
        \{\z:||\z||_2\leq 1\} & \mbox{if $||\hat{\x}_i||_2= 0$}.\end{array} \right.
\label{subgradient}
\end{align}
\item If the subgradient vectors of the optimal solution $\hat{\X}$ satisfy $||\s_i||_2<1$ for all $i\notin T(\hat{\X})$, then any optimal solution, $\check{\X}$, to \eqref{glasso} satisfies $\check{\x}_{i}=0$ for all $i\notin T(\hat{\X})$.
\item When conditions in (2) are satisfied, if in addition $\PPhi_{T(\hat{\X})}^H\PPhi_{T(\hat{\X})}\in\C^{KJ\times KJ}$ is invertible, then $\hat{\X}$ is the unique solution to \eqref{glasso}.
\end{enumerate}
\label{opunicondition}
\end{Lemma}

\begin{proof}
\leavevmode
\begin{enumerate}
\item Since problem \eqref{glasso} is convex, any optimal solution, $\hat{\X}$, must satisfy the first-order condition \eqref{firstordercondition}.
\item We first argue that when $\lambda$ is fixed, for two arbitrary different optimal solutions $\hat{\X}_1$ and $\hat{\X}_2$ to \eqref{glasso}, we have $\PPhi\cdot\text{vec}(\hat{\X}_1)=\PPhi\cdot\text{vec}(\hat{\X}_2)$. This can be proved by contradiction as follows.

    Assume $\PPhi\cdot\text{vec}(\hat{\X}_1)\neq\PPhi\cdot\text{vec}(\hat{\X}_2)$ for two arbitrary optimal solutions $\hat{\X}_1\neq\hat{\X}_2$ to \eqref{glasso}. By constructing $\hat{\X}_3=\frac{1}{2}(\hat{\X}_1+\hat{\X}_2)$, a little linear algebra yields
    \begin{equation*}
    \begin{aligned}
    \frac{1}{2}||\y-\L(\hat{\X}_3)||_2^2+\lambda||\hat{\X}_3||_{2,1}<\frac{1}{2}||\y-&\L(\hat{\X}_k)||_2^2\\
    &+\lambda||\hat{\X}_k||_{2,1}
    \end{aligned}
    \end{equation*}
    for $k\in\{1,2\}$, due to the strict convexity of the function $f(\x)=\frac{1}{2}||\y-\x||_2^2$ and the optimality of $\hat{\X}_1$ and $\hat{\X}_2$. Thus, $\hat{\X}_1$ and $\hat{\X}_2$ are not optimal. By contradiction, $\PPhi\cdot\text{vec}(\hat{\X}_1)=\PPhi\cdot\text{vec}(\hat{\X}_2)$.
    Then from \eqref{firstordercondition}, we can derive that $\s$ for different optimal solutions are the same. Therefore, assume we have an optimal solution $\hat{\X}$ such that $||\s_i||_2<1$ for all $i\notin T(\hat{\X})$, any other optimal solution, $\check{\X}$, would have subgradient vectors $||\check{\s}_i||_2=||\s_i||_2<1$ for all $i\notin T(\hat{\X})$ which implies $\check{\x}_{i}=0$ according to \eqref{subgradient}.
\item If conditions in (2) are satisfied and $\PPhi_{T(\hat{\X})}^H\PPhi_{T(\hat{\X})}\in\C^{KJ\times KJ}$ is invertible, the solution of the support restricted problem $\frac{1}{2}||\y-\PPhi_{T(\hat{\X})}\cdot\text{vec}(\X)||_2^2+\lambda||\X||_{2,1}$ is unique by solving the restricted first order condition.
\end{enumerate}
\end{proof}

\subsection{Primal-Dual Witness Construction}
The method we apply to find the conditions regarding the regularization parameter $\lambda$ and number of observations $N$ for satisfying optimality and uniqueness conditions is the primal-dual witness method \cite{wainwright2009sharp} which constructs the solution matrix, $\hat{\X}$, and subgradient vector, $\s$, through the following steps.

\begin{enumerate}
\item Conditioned on $\PPhi_{T}^H\PPhi_{T}\in\C^{KJ\times KJ}$ is invertible, we first obtain $\thX_T\in\C^{K\times J}$ by solving the support restricted problem
    \begin{equation}
    \begin{aligned}
    \thX_T=\arg\min_{\X\in\C^{K\times J}}\bigg\{\frac{1}{2}||\y-\PPhi_{T}\cdot\text{vec}&(\X)||_2^2\\
    &+\lambda||\X||_{2,1}\bigg\}.
    \label{restrictedLS}
    \end{aligned}
    \end{equation}
    The solution $\thX_T$ is unique under the invertibility condition on $\PPhi_{T}^H\PPhi_{T}$. And we set $\thX_{T^C}\in\C^{K\times(M-J)}=\bf{0}$. Thus, $\hat{\X}$ has support contained within the support $T$ of the ground truth solution $\X_0$.
\item We calculate the subgradient vector $\s_T\in\C^{JK}$ based on $\thX_T$, where $\s_T$ is a sub-vector of $\s$ consisting of $\s_j$ for all $j\in T$.
\item We solve for a vector $\s_{T^C}\in \C^{(M-J)K}$ satisfying \eqref{firstordercondition} and check whether $||\s_i||_2<1$ for all $i\notin T$.
\end{enumerate}

If $\tPPhi_T^H\tPPhi_T$ is invertible and $||\s_i||_2<1$ for all $i\in T^C$, $\hat{\X}$ constructed via the primal-dual witness method is the unique optimal solution to \eqref{glasso} with its support contained within the support of the ground truth solution $\X_0$. And note that the primal-dual witness construction succeeds only if the problem \eqref{glasso} has a unique solution whose support is contained within the support of the ground truth. The challenges of the construction lie in characterizing the regularization parameter $\lambda$ and the number of observations $N$ such that $||\s_i||_2<1$ for all $i\in T^C$.

To simplify the notation, without loss of generality, we assume the support of $\X_0$ is the first $J$ columns and $T=\{1, 2, ..., J\}$ throughout the proof. Therefore, rewriting \eqref{firstordercondition} into matrix multiplication form results in
\begin{equation}
\begin{aligned}
\begin{bmatrix}
\tPPhi_T^H\tPPhi_T \quad \tPPhi_T^H\tPPhi_{T^C}\\
\tPPhi_{T^C}^H\tPPhi_T \quad \tPPhi_{T^C}^H\tPPhi_{T^C}
\end{bmatrix}&
\begin{bmatrix}
\text{vec}(\thX_T)-\text{vec}(\X_{0,T})\\
\bf{0}
\end{bmatrix}
\\&-
\begin{bmatrix}
\tPPhi_T^H\\
\tPPhi_{T^C}^H
\end{bmatrix}
\n+\lambda
\begin{bmatrix}
\ts_T\\
\ts_{T^C}
\end{bmatrix}=0.
\end{aligned}
\label{bigmatrix}
\end{equation}
When $\tPPhi_T^H\tPPhi_T$ is invertible, from \eqref{bigmatrix} we can derive that
\begin{align*}
\Delta(\X)=\text{vec}(\thX_T)-\text{vec}(\X_{0,T})=(\tPPhi_T^H\tPPhi_T)^{-1}\left(\tPPhi_T^H\n-\lambda\ts_T\right)
\end{align*}
and
\begin{align}
\ts_{T^C}=\frac{1}{\lambda}\left( \tPPhi_{T^C}^H\n- \tPPhi_{T^C}^H\tPPhi_T\Delta(\X)\right).
\label{stc}
\end{align}
Substituting the full expression of $\Delta(\X)$ into \eqref{stc} results in

\begin{equation}
\begin{aligned}
\ts_{T^C} = \tPPhi^H_{T^C}&\left(\I_N-\tPPhi_T(\tPPhi_T^H\tPPhi_T)^{-1}\tPPhi_T^H\right)\frac{\n}{\lambda}\\
&\qquad\qquad+\tPPhi_{T^C}^H\tPPhi_T(\tPPhi_T^H\tPPhi_T)^{-1}\ts_T.
\end{aligned}
\label{stcnew}
\end{equation}

\subsection{Important Lemmas}
In this section, we introduce some important lemmas and propositions that will be applied during the proof of Theorem \ref{maintheorem}.
First is the isometry bound for the linear operator $\L$ defined in \eqref{L} (and $\PPhi$ defined in \eqref{phi}) which can be found in Lemma 4.3 in \cite{ling2015self}.

\begin{Lemma}\cite[Lemma 4.3]{ling2015self}
(Isometry) For the linear operator $\L$ defined in (\ref{L}) with $\B^H\B=\I_K$ and $\delta>0$,
\begin{align*}
||\PPhi^H_T\PPhi_T-\I_{T}||=||\L_{T}^*\L_{T}-\I_T||\leq\delta
\end{align*}
with probability at least $1-N^{-\alpha+1}$ where $\I_T$ is the identity operator on the support $T$ such that $\I_T(\X)=\X_T$, if $\A$ is a random Gaussian matrix and $N\geq C_{\alpha}\mu_{max}^2KJ\max\{\log(N)/\delta^2,\log^2(N)/\delta\}$. Here $C_{\alpha}$ is a constant that grows linearly with $\alpha>1$.
\label{isometryproperty}
\end{Lemma}
According to Lemma A.12 in \cite{foucart2013mathematical}, if $||\PPhi^H_T\PPhi_T-\I_{T}||\leq\delta<1$, $\PPhi^H_T\PPhi_T$ is invertible and
$||(\PPhi^H_T\PPhi_T)^{-1}||\leq (1-\delta)^{-1}$. In addition, we have the following quadratic Gaussian tail bound proposition, developed from Theorem 1 in \cite{hsu2012tail}.
\begin{Proposition}
Let $\H\in\C^{K\times N}$ and $\mathbf{\Sigma}=\H^H\H$. Let $\a\in\C^{N}$ whose real and imaginary entries follow the i.i.d normal distribution with 0 mean and $\sigma^2$ variance. For all $\alpha>0$,
\begin{align*}
\text{Pr}\left( ||\H\a||_2^2>\sigma^2\left[2\tr(\mathbf{\Sigma})+2\sqrt{2\tr(\mathbf{\Sigma}^2)\alpha}+2||\mathbf{\Sigma}||\alpha \right]\right)\\
\leq e^{-\alpha}.
\end{align*}
If $\a\in\R^N$ only contains the real part, for all $\alpha>0$,
\begin{equation*}
\begin{aligned}
\text{Pr}\left( ||\H\a||_2^2>\sigma^2\left[\tr(\mathbf{\Sigma})+2\sqrt{\tr(\mathbf{\Sigma}^2)\alpha}+2||\mathbf{\Sigma}||\alpha \right]\right)\\
\leq e^{-\alpha}.
\end{aligned}
\end{equation*}

\label{quadraticgaussian}
\end{Proposition}

\begin{proof}
When $\H$ and $\a$ are a complex matrix and vector, we can write $\H=\H_R+i\H_I$ and $\a=\a_R+i\a_I$ where $\H_R$, $\H_I$, $\a_R$ and $\a_I$ are all real and the entries of $\a_R$ and $\a_I$ are i.i.d Gaussian random variables with 0 mean and $\sigma^2$ variance. We then have
\begin{equation*}
\begin{aligned}
||\H\a||_2^2&=||(\H_R+i\H_I)(\a_R+i\a_I)||_2^2\\
&=||(\H_R\a_R-\H_I\a_I)+i(\H_R\a_I+\H_I\a_R)||_2^2\\
&=||\H_R\a_R-\H_I\a_I||_2^2+||\H_R\a_I+\H_I\a_R||_2^2\\
&=\left|\left|
\begin{bmatrix}
\H_R&-\H_I\\
\H_I&\H_R
\end{bmatrix}
\begin{bmatrix}
\a_R\\
\a_I
\end{bmatrix}
\right|\right|_2^2.
\end{aligned}
\end{equation*}
Define $\H_o=\begin{bmatrix}
\H_R&-\H_I\\
\H_I&\H_R
\end{bmatrix}$ and $\mathbf{\Sigma}_o=\H_o^T\H_o$. $\mathbf{\Sigma}_o$ has the form
\begin{equation}
\begin{aligned}
\mathbf{\Sigma}_o&=\begin{bmatrix}
\H_R^T&\H_I^T\\
-\H_I^T&\H_R^T
\end{bmatrix}\begin{bmatrix}
\H_R&-\H_I\\
\H_I&\H_R
\end{bmatrix}\\
&=
\begin{bmatrix}
\H_R^T\H_R+\H_I^T\H_I&-\H_R^T\H_I+\H_I^T\H_R\\
-\H_I^T\H_R+\H_R^T\H_I&\H_I^T\H_I+\H_R^T\H_R
\end{bmatrix}\\
&=
\begin{bmatrix}
\H_1&-\H_2\\
\H_2&\H_1
\end{bmatrix}
\end{aligned}
\label{sigmao}
\end{equation}
where we define $\H_1=\H_R^T\H_R+\H_I^T\H_I$ and $\H_2=-\H_I^T\H_R+\H_R^T\H_I$.
Applying Theorem 1 in \cite{hsu2012tail}, we get
\begin{equation*}
\begin{aligned}
\text{Pr}\left(||\H\a||_2^2>\sigma^2\left[\tr(\mathbf{\Sigma}_o)+2\sqrt{\tr(\mathbf{\Sigma}_o^2)\alpha}+2||\mathbf{\Sigma}_o||\alpha\right] \right)\\
\leq e^{-\alpha}.
\end{aligned}
\end{equation*}
If we further define
\begin{equation}
\begin{aligned}
\mathbf{\Sigma}&=\H^H\H=(\H_R+i\H_I)^H(\H_R+i\H_I)\\
&=(\H_R^T-i\H_I^T)(\H_R+i\H_I)\\
&=(\H_R^T\H_R+\H_I^T\H_I)+i(-\H_I^T\H_R+\H_R^T\H_I)\\
&=\H_1+i\H_2,
\end{aligned}
\label{sigma}
\end{equation}
by comparing \eqref{sigmao} and \eqref{sigma}, one can check that $\tr(\mathbf{\Sigma}_o)=2\tr(\H_1)=2\tr(\mathbf{\Sigma})$ since $\tr(\H_2)=0$, $\tr(\mathbf{\Sigma}_o^2)=||\mathbf{\Sigma}_o||_F^2=2(||\H_1||_F^2+||\H_2||_F^2)=2||\mathbf{\Sigma}||_F^2=2\tr(\mathbf{\Sigma}^2)$, and
\begin{equation*}
\begin{aligned}
&\quad||\mathbf{\Sigma}_o||=\left|\left|\begin{bmatrix}
\H_1&-\H_2\\
\H_2&\H_1
\end{bmatrix}\right|\right|\\
&=\max_{\left|\left|\begin{bmatrix}\x_1\\ \x_2\end{bmatrix}\right|\right|_2=1}\left|\left|\begin{bmatrix}
\H_1&-\H_2\\
\H_2&\H_1
\end{bmatrix}\begin{bmatrix}\x_1\\ \x_2\end{bmatrix}\right|\right|_2\\
&=\max_{\left|\left|\begin{bmatrix}\x_1\\ \x_2\end{bmatrix}\right|\right|_2=1}\sqrt{ ||\H_1\x_1-\H_2\x_2||_2^2+||\H_2\x_1+\H_1x_2||_2^2 }\\
&=\max_{||\x_1||_2^2+||\x_2||^2_2=1}\sqrt{ ||(\H_1\x_1-\H_2\x_2)+i(\H_2\x_1+\H_1\x_2)||_2^2 }\\
&=\max_{||\x_1||_2^2+||\x_2||^2_2=1}\sqrt{ ||(\H_1+i\H_2)(\x_1+i\x_2)||_2^2 }\\
&=\max_{||\x_1||_2^2+||\x_2||^2_2=1} ||\mathbf{\Sigma}(\x_1+i\x_2)||_2=||\mathbf{\Sigma}||
\end{aligned}
\end{equation*}
where $\x_1$ and $\x_2\in\R^{N}$ since $\mathbf{\Sigma}_o$ is a real matrix, so that the vector corresponding to its largest singular value is also real. Therefore, we have
\begin{equation*}
\begin{aligned}
\text{Pr}\left(||\H\a||_2^2>\sigma^2\left[2\tr(\mathbf{\Sigma})+2\sqrt{2\tr(\mathbf{\Sigma}^2)\alpha}+2||\mathbf{\Sigma}||\alpha\right] \right)\\
\leq e^{-\alpha}.
\end{aligned}
\end{equation*}
Similarly, when $\a$ only contains the real part
\begin{align*}
||\H\a||_2^2=\left|\left|\begin{bmatrix}
\H_R\\
\H_I
\end{bmatrix}\a_{R}
\right|\right|_2^2,
\end{align*}
$\mathbf{\Sigma}$ still follows \eqref{sigma} and
\begin{align*}
\mathbf{\Sigma}_o= [\H_R^T\H_R+\H_I^T\H_I]=\H_1.
\end{align*}
In this case, $\tr(\mathbf{\Sigma}_o)=\tr(\mathbf{\Sigma})$, $\tr(\mathbf{\Sigma}_o^2)\leq\tr(\mathbf{\Sigma}^2)$ and since $\mathbf{\Sigma}_o$ is real, we have
\begin{equation*}
\begin{aligned}
&\quad||\mathbf{\Sigma}_o||=\max_{||\x||_2=1,\x\in\R^N}\sqrt{||\H_1\x||_2^2}\\
&\leq \max_{||\x||_2=1,\x\in\R^N}\sqrt{||\H_1\x||_2^2+||\H_2\x||_2^2}\\
&=\max_{||\x||_2=1,\x\in\R^N}\sqrt{||\H_1\x+i\H_2\x||_2^2}\\
&\leq\max_{||\x||_2=1,\x\in\C^N}\sqrt{||(\H_1+i\H_2)\x||_2^2}\\
&=\max_{||\x||_2=1,\x\in\C^N}||\mathbf{\Sigma}\x||_2=||\mathbf{\Sigma}||.
\end{aligned}
\end{equation*}
So we have, for $\alpha>0$,
\begin{equation*}
\begin{aligned}
\tr(\mathbf{\Sigma})+2\sqrt{\tr(\mathbf{\Sigma}^2)\alpha}&+2||\mathbf{\Sigma}||\alpha\geq\\
&\tr(\mathbf{\Sigma}_o)+2\sqrt{\tr(\mathbf{\Sigma}_o^2)\alpha}+2||\mathbf{\Sigma}_o||\alpha
\end{aligned}
\end{equation*}
which results in
\begin{equation*}
\begin{aligned}
\text{Pr}\left( ||\H\a||_2^2>\sigma^2\left[\tr(\mathbf{\Sigma})+2\sqrt{\tr(\mathbf{\Sigma}^2)\alpha}+2||\mathbf{\Sigma}||\alpha \right]\right)\\
\leq e^{-\alpha}.
\end{aligned}
\end{equation*}

\end{proof}

\begin{Proposition}
Let $\H\in\C^{K\times N}$ and $\mathbf{\Sigma}=\H^H\H$. Let $\a\in\C^{N}$ whose real and imaginary entries follow the i.i.d normal distribution with 0 mean and $\sigma^2$ variance. For all $\alpha>1$,
\begin{align*}
\text{Pr}\left( ||\H\a||_2^2>\sigma^2\left[2+(2\sqrt{2}+2)\alpha\right]\tr(\mathbf{\Sigma})\right)\leq e^{-\alpha}.
\end{align*}
If $\a\in\R^N$ only contains the real part, for all $\alpha>1$,
\begin{align*}
\text{Pr}\left( ||\H\a||_2^2>\sigma^2\left(1+4\alpha\right)\tr(\mathbf{\Sigma})\right)\leq e^{-\alpha}.
\end{align*}

\label{genquagaussian}
\end{Proposition}

\begin{proof}
Since $\mathbf{\Sigma}$ is a positive semi-definite and hermitian matrix, all its eigenvalues, $\lambda_i$, are non-negative. Thus, $\tr(\mathbf{\Sigma}^2)=\sum_{i=1}^N\lambda_i^2\leq(\sum_{i=1}^N\lambda_i)^2=\tr(\mathbf{\Sigma})^2$ and $||\mathbf{\Sigma}||=\lambda_{max}\leq\sum_{i=1}^N\lambda_i=\tr(\mathbf{\Sigma})$.
As a result, for $\alpha>1$,
\begin{equation*}
\begin{aligned}
&\sigma^2\left[2+(2\sqrt{2}+2)\alpha\right]\tr(\mathbf{\Sigma})\\
&\qquad\geq\sigma^2\left[2\tr(\mathbf{\Sigma})+2\sqrt{2\tr(\mathbf{\Sigma}^2)\alpha}+2||\mathbf{\Sigma}||\alpha \right]
\end{aligned}
\end{equation*}
and
\begin{equation*}
\begin{aligned}
&\sigma^2\left(1+4\alpha\right)\tr(\mathbf{\Sigma})\\
&\qquad\geq\sigma^2\left[\tr(\mathbf{\Sigma})+2\sqrt{\tr(\mathbf{\Sigma}^2)\alpha}+2||\mathbf{\Sigma}||\alpha \right].
\end{aligned}
\end{equation*}
Then applying Proposition \ref{quadraticgaussian} yields Proposition \ref{genquagaussian}.
\end{proof}

\subsection{Bounding $||\ts_{T^C}||_{2,\infty}$}
Recalling \eqref{stcnew}, to prove that $||\s_i||_2<1$ for all $i\in T^C$ which is equivalent to $||\ts_{T^C}||_{2,\infty}<1$, where the $\ell_{2,\infty}$-norm of the subgradient vector is defined in Section \ref{notation}, we only need to show that for a constant $\gamma\in(0,1)$,

\begin{align*}
||\tPPhi^H_{T^C}\left(\I_N-\tPPhi_T(\tPPhi_T^H\tPPhi_T)^{-1}\tPPhi_T^H\right)\frac{\n}{\lambda}||_{2,\infty}\leq\frac{\gamma}{2}
\end{align*}
and
\begin{align*}
||\tPPhi_{T^C}^H\tPPhi_T(\tPPhi_T^H\tPPhi_T)^{-1}\ts_T||_{2,\infty}\leq\frac{\gamma}{2}.
\end{align*}
Then by the triangle inequality, $||\ts_{T^C}||_{2,\infty}\leq\gamma<1$.

\begin{Lemma}
Conditioned on $\tPPhi_T^H\tPPhi_T$ being invertible, we have \begin{align*}
||\tPPhi^H_{T^C}\left(\I_N-\tPPhi_T(\tPPhi_T^H\tPPhi_T)^{-1}\tPPhi_T^H\right)\frac{\n}{\lambda}||_{2,\infty}\leq\frac{\gamma}{2}
\end{align*}
for $\gamma\in(0,1)$ with probability at least $1-N^{-\alpha+1}$ when
\begin{align*}
\lambda\geq\sqrt{\frac{C_{\alpha}\sigma^2\mu_{max}^2K[\log(M-J)+\log(N)]}{\gamma^2}}
\end{align*}
and
\begin{align*}
N\geq 10\log(M-J)+10\alpha\log(N)
\end{align*}
where $C_{\alpha}$ is a constant that grows linearly with $\alpha>1$.
\label{stc1}
\end{Lemma}

\begin{proof}
$||\tPPhi^H_{T^C}\left(\I_N-\tPPhi_T(\tPPhi_T^H\tPPhi_T)^{-1}\tPPhi_T^H\right)\frac{\n}{\lambda}||_{2,\infty}\leq\frac{\gamma}{2}$ for $\gamma\in(0,1)$ is equivalent to
\begin{align*}
\max_{i\in T^{C}}||\tPPhi^H_{i}\left(\I_N-\tPPhi_T(\tPPhi_T^H\tPPhi_T)^{-1}\tPPhi_T^H\right)\n||^2_{2}\leq\frac{\lambda^2\gamma^2}{4}
\end{align*}
where $\PPhi_{i}$ ($i\in T^C$) is the sub-matrix containing the $[K(i-1)+1]$ to $[K(i-1)+K]$-th columns of $\PPhi$. If we define $\H_i=\tPPhi^H_{i}\left(\I_N-\tPPhi_T(\tPPhi_T^H\tPPhi_T)^{-1}\tPPhi_T^H\right)$, the projection matrix $\P=\left(\I_N-\tPPhi_T(\tPPhi_T^H\tPPhi_T)^{-1}\tPPhi_T^H\right)$, and $\mathbf{\Sigma}=\H_i^H\H_i$, we get
\begin{equation*}
\begin{aligned}
\tr(\mathbf{\Sigma})&=||\H_i||^2_F=||\tPPhi^H_{i}\cdot\P||_F^2=||\B^H\diag(\a_i)^H\cdot\P||_F^2\\
&=||\P\diag(\a_i)\B||_F^2\\
&=||\P\left[ \diag(\b_1)\a_i, \diag(\b_2)\a_i,\cdots, \diag(\b_K)\a_i \right]||_F^2\\
&=\sum_{k=1}^K||\P\diag(\b_k)\a_i||^2_2\leq\sum_{k=1}^K||\P||^2\frac{\mu_{max}^2}{N}||\a_i||^2_2\\
&\leq\frac{\mu_{max}^2K}{N}||\a_i||_2^2.
\end{aligned}
\end{equation*}
Since $\n$ is the additive Gaussian noise vector, applying Proposition \ref{genquagaussian} gives us, for $\alpha_1>1$
\begin{equation}
\begin{aligned}
\text{Pr}\left( ||\H_i\n||_2^2>\frac{\lambda^2\gamma^2}{4}\geq\sigma^2\left[2+(2\sqrt{2}+2)\alpha_1\right]\tr(\mathbf{\Sigma})\right)\leq \\e^{-\alpha_1},
\label{Dn}
\end{aligned}
\end{equation}
in which we need
\begin{equation}
\begin{aligned}
\lambda&\geq\sqrt{\frac{\sigma^2\left[8+(8\sqrt{2}+8)\alpha_1\right]\mu_{max}^2K||\a_i||_2^2}{\gamma^2N}}\\
&\geq\sqrt{\frac{\sigma^2\left[8+(8\sqrt{2}+8)\alpha_1\right]\tr(\mathbf{\Sigma})}{\gamma^2}}.
\end{aligned}
\label{lambda}
\end{equation}
To control the term $||\a_i||_2^2$, we define an event $E=\{\max_{i\in T^{C}}||\a_i||_2^2< 2N\}$. Because each entry of $\a_i\in\R^N$ follows the standard normal distribution, $||\a_i||_2^2$ is a $\chi_N^2$ random variable. According to Lemma 1 in \cite{laurent2000adaptive}, for $\alpha_2>0$
\begin{align*}
\text{Pr}(||\a_i||_2^2\geq2\sqrt{\alpha_2N}+2\alpha_2+N)\leq e^{-\alpha_2}.
\end{align*}
By solving $2N\geq2\sqrt{\alpha_2N}+2\alpha_2+N$, we require $\alpha_2\leq(\frac{2\sqrt{3}-2}{4})^2N\approx 0.1340N$. So for $0<\alpha_2\leq\frac{N}{10}$, we have
\begin{align*}
\text{Pr}(||\a_i||_2^2\geq2N)\leq e^{-\alpha_2}.
\end{align*}
Taking the union over all $i\in T^C$ gives us
\begin{align*}
\text{Pr}(E^C)\leq (M-J)e^{-\alpha_2}
\end{align*}
which is meaningful when $\log(M-J)\leq\alpha_2\leq\frac{N}{10}$.

In addition, if we define another event $F=\{ \max_{i\in T^{C}}||\H_i\n||_2^2>\frac{\lambda^2\gamma^2}{4} \}$, conditioned on $E$ and with
\begin{align}
\lambda\geq\sqrt{\frac{\sigma^2\left[16+(16\sqrt{2}+16)\alpha_1\right]\mu_{max}^2K}{\gamma^2}},
\label{lambda2}
\end{align}
by taking the union of \eqref{Dn} over all $i\in T^C$, we obtain
\begin{align*}
\text{Pr}(F\mid E)\leq (M-J)e^{-\alpha_1}.
\end{align*}
Therefore,
\begin{equation*}
\begin{aligned}
\text{Pr}(F\mid E)+\text{Pr}(E^C)&\leq (M-J)e^{-\alpha_1} +(M-J)e^{-\alpha_2}\\
            &= 2N^{-\alpha}\leq N^{-\alpha+1}
\end{aligned}
\end{equation*}
for $\alpha>1$ by setting $\alpha_1=\alpha_2=\log(M-J)+\alpha\log(N)$. Substituting $\alpha_1$ into \eqref{lambda2} and some simplification yields
\begin{equation*}
\begin{aligned}
\lambda\geq\sqrt{\frac{C_{\alpha}\sigma^2\mu_{max}^2K[\log(M-J)+\log(N)]}{\gamma^2}}
\end{aligned}
\end{equation*}
where $C_{\alpha}=(16\sqrt{2}+16)\alpha+16$ is a constant that grows linearly with $\alpha>1$. Moreover, $\log(M-J)\leq\alpha_2=\log(M-J)+\alpha\log(N)\leq \frac{N}{10}$ requires $N\geq 10\log(M-J)+10\alpha\log(N)$.
Finally, the law of probability implies
\begin{equation*}
\begin{aligned}
&\quad \text{Pr}(\max_{i\in T^{C}}||\tPPhi^H_{i}\left(\I_N-\tPPhi_T(\tPPhi_T^H\tPPhi_T)^{-1}\tPPhi_T^H\right)\n||^2_{2}\leq\frac{\lambda^2\gamma^2}{4})\\
&=\text{Pr}(F^C)\geq\text{Pr}(F^C\cap E)=1-[\text{Pr}(E^C)+\text{Pr}(F\cap E)]\\
&\geq 1-[\text{Pr}(E^C)+\text{Pr}(F| E)]\geq 1-N^{-\alpha+1}.
\end{aligned}
\end{equation*}
\end{proof}

\begin{Lemma}
Conditioned on $\tPPhi_T^H\tPPhi_T$ being invertible and $||(\tPPhi_T^H\tPPhi_T)^{-1}||\leq2$, we have
\begin{align*}
||\tPPhi_{T^C}^H\tPPhi_T(\tPPhi_T^H\tPPhi_T)^{-1}\ts_T||_{2,\infty}\leq\frac{\gamma}{2}
\end{align*}
for $\gamma\in(0,1)$ with probability at least $1-N^{-\alpha}$ when
\begin{align*}
N\geq C_{\alpha}\frac{\mu_{max}^2KJ}{\gamma^2}[\log(M-J)+\log(N)],
\end{align*}
where $C_{\alpha}$ is a constant that grows linearly with $\alpha>1$.
\label{stc2}
\end{Lemma}

\begin{proof}
$||\tPPhi_{T^C}^H\tPPhi_T(\tPPhi_T^H\tPPhi_T)^{-1}\ts_T||_{2,\infty}\leq\frac{\gamma}{2}$ for $\gamma\in(0,1)$ can be reformulated as
\begin{equation*}
\begin{aligned}
&\quad\max_{i\in T^C}||\tPPhi_{i}^H\tPPhi_T(\tPPhi_T^H\tPPhi_T)^{-1}\ts_T||^2_{2}\\
&=\max_{i\in T^C}\left|\left|
\begin{bmatrix}
\a_i^H\diag(\b_1)^H\\
\a_i^H\diag(\b_2)^H\\
\vdots\\
\a_i^H\diag(\b_K)^H\\
\end{bmatrix}
\cdot\v
 \right|\right|^2_2\\
&=\max_{i\in T^C}||\v^H[\diag(\b_1)\a_i,\diag(\b_2)\a_i,\cdots,\diag(\b_K)\a_i ] ||^2_2\\
&=\max_{i\in T^C}\left|\left|
\begin{bmatrix}
\v^H\diag(\b_1)\\
\v^H\diag(\b_2)\\
\vdots\\
\v^H\diag(\b_K)
\end{bmatrix}
\a_i\right|\right|^2_2=\max_{i\in T^C}||\H\a_i||^2_2
\leq\frac{\gamma^2}{4}
\end{aligned}
\end{equation*}
where we define $\v=\tPPhi_T(\tPPhi_T^H\tPPhi_T)^{-1}\ts_T\in\C^{N}$ and
\begin{align*}
\H=\begin{bmatrix}
\v^H\diag(\b_1)\\
\v^H\diag(\b_2)\\
\vdots\\
\v^H\diag(\b_K)
\end{bmatrix}\in\C^{K\times N}.
\end{align*}
Therefore, for $\mathbf{\Sigma} = \H^H\H$, we have
\begin{equation}
\begin{aligned}
\tr(\mathbf{\Sigma})=\left|\left|\begin{bmatrix}
\v^H\diag(\b_1)\\
\v^H\diag(\b_2)\\
\vdots\\
\v^H\diag(\b_K)
\end{bmatrix}\right|\right|_F^2&\leq\frac{\mu_{max}^2K}{N}||\v||_2^2\leq \frac{2\mu_{max}^2KJ}{N}
\end{aligned}
\label{tr}
\end{equation}
since
\begin{equation*}
\begin{aligned}
||\v||_2^2&=|\v^H\v|=|\ts_T^H(\tPPhi_T^H\tPPhi_T)^{-1}\tPPhi_T^H\tPPhi_T(\tPPhi_T^H\tPPhi_T)^{-1}\ts_T|\\
&=|\ts_T^H(\tPPhi_T^H\tPPhi_T)^{-1}\ts_T|\leq||(\tPPhi_T^H\tPPhi_T)^{-1}||\cdot||\ts_T||^2_2\leq 2J.
\end{aligned}
\end{equation*}
Because $\a_i\in\R^N$ for $i\in T^C$ is independent of $\tPPhi_T$ and $\a_i$'s entries follow the i.i.d standard normal distribution, Proposition \ref{genquagaussian} implies, for $\alpha_1>1$
\begin{align*}
\text{Pr}\left(||\H\a_i||_2^2>\left(1+4\alpha_1\right)\tr(\mathbf{\Sigma})\right)\leq e^{-\alpha_1}.
\end{align*}
To ensure that $\frac{\gamma^2}{4}\geq\left(1+4\alpha_1\right)\tr(\mathbf{\Sigma})$ , we need
\begin{align}
N\geq\frac{(8+32\alpha_1)\mu_{max}^2KJ}{\gamma^2}.
\label{stc2N}
\end{align}
By taking the union over all $i\in T^C$, we get
\begin{equation*}
\begin{aligned}
\text{Pr}\left(\max_{i\in T^C}||\tPPhi_{i}^H\tPPhi_T(\tPPhi_T^H\tPPhi_T)^{-1}\ts_T||^2_{2}>\frac{\gamma^2}{4}\right)&\leq(M-J)e^{-\alpha_1}\\
&= N^{-\alpha}
\end{aligned}
\end{equation*}
if we set $\alpha_1= \log(M-J)+\alpha\log(N)$ for $\alpha>1$. Substituting the full expression of $\alpha_1$ into \eqref{stc2N} and some simplification yields
\begin{align*}
N\geq C_{\alpha}\frac{\mu_{max}^2KJ}{\gamma^2}[\log(M-J)+\log(N)]
\end{align*}
where $C_{\alpha}=32\alpha+8$ is a constant that grows linearly with $\alpha>1$.
\end{proof}

\subsection{Bounding $||\thX_T-\X_{0,T} ||_{2,\infty}$}
When the support of the unique optimal solution $\hat{\X}$ is contained within the support of the ground truth solution $\X_0$, the recovery error $||\hat{\X}-\X_0||_{2,\infty}=||\thX_T-\X_{0,T} ||_{2,\infty}$. And because the optimal solution on the support, $\thX_T\in \C^{K\times J}$ (we assume $\thX_T\neq\X_{0,T}$, otherwise $||\X_{0,T}- \thX_T ||_{2,\infty}=0$) is attained by solving the support-restricted regularized least square problem \eqref{restrictedLS} whose objective function $f(\text{vec}(\X))=\frac{1}{2}||\y-\PPhi_{T}\cdot\text{vec}(\X)||_2^2+\lambda||\X||_{2,1}$ is strongly convex, since $\frac{1}{2}||\y-\PPhi_{T}\cdot\text{vec}(\X)||_2^2$ is strongly convex conditioned on $\PPhi_{T}^H\PPhi_{T}$ being positive definite and $\lambda||\X||_{2,1}$ is convex, by the property of the strongly convex function, we have
\begin{equation*}
\begin{aligned}
&f(\text{vec}(\thX_T))\geq f(\text{vec}(\X_{0,T}))
+\Re\Big\{ g(\text{vec}(\X_{0,T}))^H\cdot\\
&\qquad\qquad\left[\text{vec}(\thX_T)-\text{vec}(\X_{0,T})\right]\Big\}
+\frac{m}{2}||\thX_T-\X_{0,T}||_F^2
\end{aligned}
\end{equation*}
where $g(\text{vec}(\X_{0,T}))$ is the subgradient of $f(\text{vec}(\X_{0,T}))$. In addition, if we set $\delta=\frac{1}{2}$ in Lemma \ref{isometryproperty}, we have $||(\PPhi^H_T\PPhi_T)^{-1}||\leq 2$ according to Lemma A.12 in \cite{foucart2013mathematical}, which implies $\PPhi_{T}^H\PPhi_{T}\succeq \frac{1}{2}\I$. As a result, $m=\frac{1}{2}$.
Then by the H\"older inequality,
\begin{equation}
\begin{aligned}
f(\text{vec}(\thX_T))&\geq f(\text{vec}(\X_{0,T}))+\Re \Big\{g(\text{vec}(\X_{0,T}))^H\cdot\\
&\quad\left[\text{vec}(\thX_T)-\text{vec}(\X_{0,T})\right]\Big\} +\frac{1}{4}||\thX_T-\X_{0,T}||_F^2\\
&\geq f(\text{vec}(\X_{0,T}))-||g(\text{vec}(\X_{0,T}))||_{2,\infty}\cdot\\
&\quad||\thX_T-\X_{0,T}||_{2,1}
+\frac{1}{4}||\thX_T-\X_{0,T}||_F^2\\
&\geq f(\text{vec}(\X_{0,T}))-||g(\text{vec}(\X_{0,T}))||_{2,\infty}\cdot\\
&\quad||\thX_T-\X_{0,T}||_{2,1}+ \\
&\quad\frac{1}{4\sqrt{J}}||\thX_T-\X_{0,T}||_{2,\infty}||\thX_T-\X_{0,T}||_{2,1}\\
\label{stronglyconvex}
\end{aligned}
\end{equation}
where the $\ell_{2,\infty}$-norm of the subgradient vector is defined in Section \ref{notation}, and the third inequality comes from the fact that $\frac{||\X||_F^2}{||\X||_{2,1}||\X||_{2,\infty}}\geq \frac{1}{\sqrt{J}}$. Because if $||\X||_F^2=L$, one can check that $||\X||_{2,\infty}\leq \sqrt{L}$ and $||\X||_{2,1}\leq \sqrt{LJ}$ where the equality is achieved when the 2-norm of all $J$ columns are the same.

Therefore, since $\thX_T\neq\X_{0,T}$ and $\thX_T$ is the optimal solution, $f(\text{vec}(\thX_T))\leq f(\text{vec}(\X_{0,T}))$, \eqref{stronglyconvex} yields
\begin{equation}
\begin{aligned}
&\quad||\thX_T-\X_{0,T}||_{2,\infty}\leq 4\sqrt{J}||g(\text{vec}(\X_{0,T}))||_{2,\infty}\\
&=4\sqrt{J}||\PPhi_{T}^H\left[\PPhi_{T}\text{vec}(\X_{0,T})-\y\right]+\lambda \ts_{0,T}||_{2,\infty}\\
&=4\sqrt{J}||-\PPhi_{T}^H\n+\lambda \ts_{0,T}||_{2,\infty}\\
&\leq4\sqrt{J}\left(||\PPhi_{T}^H\n||_{2,\infty}+||\lambda \ts_{0,T}||_{2,\infty}\right)\\
&=4\sqrt{J}\left(||\PPhi_{T}^H\n||_{2,\infty}+\lambda\right)
\end{aligned}
\label{2infdeltax}
\end{equation}
where we have used $\y=\PPhi_{T}\text{vec}(\X_{0,T})+\n$ and $||\ts_{0,T}||_{2,\infty}=1$. Now we turn to bound the term $||\PPhi_{T}^H\n||_{2,\infty}$ applying the following lemma.

\begin{Lemma}
Conditioned on $\tPPhi_T^H\tPPhi_T$ being invertible, we have
\begin{align}
||\PPhi_{T}^H\n||_{2,\infty}\leq\sqrt{C_{\alpha}\sigma^2\mu_{max}^2K\left[\log(J)+\log(N)\right]}
\label{phinbound}
\end{align}
with probability at least $1-N^{-\alpha+1}$ when
\begin{align*}
N\geq 10\log(J)+10\alpha\log(N)
\end{align*}
where $C_{\alpha}$ is a constant that grows linearly with $\alpha>1$.
\label{phinlemma}
\end{Lemma}

\begin{proof}
If we define $\PPhi_{j}$ ($j\in T$) to be the $[K(j-1)+1]$ to $[K(j-1)+K]$-th columns of $\PPhi$, we have $||\PPhi_{T}^H\n||_{2,\infty}=\max_{j\in T}||\PPhi_{j}^H\n||_2$. For an arbitrary $j\in T$, let
$\mathbf{\Sigma}=\PPhi_{j}\PPhi_{j}^H$,
\begin{equation*}
\begin{aligned}
\tr(\mathbf{\Sigma})&=\tr(\PPhi_{j}\PPhi_{j}^H)=\tr(\PPhi_{j}^H\PPhi_{j})\\
&=||\PPhi_{j}||_F^2=||\left[\diag(\b_1)\a_j,\cdots,\diag(\b_K)\a_j\right]||_F^2\\
&=\sum_{k=1}^K ||\diag(\b_k)\a_j||_2^2\leq \frac{\mu_{max}^2K}{N}||\a_j||_2^2.
\end{aligned}
\end{equation*}
If we define an event $E=\{\max_{j\in T}||\a_j||_2^2<2N\}$, in the proof of Lemma \ref{stc1} we have shown that for $0<\alpha_1\leq\frac{N}{10}$
\begin{align*}
\text{Pr}(||\a_i||_2^2\geq2N)\leq e^{-\alpha_1}.
\end{align*}
Taking the union over all $j\in T$ results in
\begin{align*}
\text{Pr}(E^C)\leq Je^{-\alpha_1}
\end{align*}
which is meaningful when $\log(J)\leq\alpha_1\leq\frac{N}{10}$. Therefore, conditioned on $E$, $\tr(\mathbf{\Sigma})<2\mu_{max}^2K$. Applying Proposition \ref{genquagaussian} gives us, for $\alpha_2>1$
\begin{equation*}
\begin{aligned}
\text{Pr}(||\PPhi_{j}^H\n||_2^2>\left[4+(4\sqrt{2}+4)\alpha_2\right]\sigma^2\mu_{max}^2K\mid E)\leq e^{-\alpha_2}.
\end{aligned}
\end{equation*}
Taking the union over all $j\in T$ yields
\begin{equation}
\begin{aligned}
\text{Pr}(\max_{j\in T}||\PPhi_{j}^H\n||_2^2>\left[4+(4\sqrt{2}+4)\alpha_2\right]\sigma^2\mu_{max}^2K\mid E)\leq \\Je^{-\alpha_2}.
\end{aligned}
\label{phin}
\end{equation}
Therefore,
\begin{equation*}
\begin{aligned}
\text{Pr}(\max_{j\in T}||\PPhi_{j}^H\n||_2^2>\left[4+(4\sqrt{2}+4)\alpha_2\right]\sigma^2\mu_{max}^2K\mid E)\\
+\text{Pr}(E^C)\leq Je^{-\alpha_2}+Je^{-\alpha_1}=2N^{-\alpha}\leq N^{-\alpha+1}
\end{aligned}
\end{equation*}
if we set $\alpha_1=\alpha_2=\log(J)+\alpha\log(N)$ for $\alpha>1$. Moreover, $\log(J)\leq\alpha_1=\log(J)+\alpha\log(N)\leq\frac{N}{10}$ requires that $N\geq 10\log(J)+10\alpha\log(N)$. Substituting $\alpha_2=\log(J)+\alpha\log(N)$ into \eqref{phin} and some simplification yields that, for an event
\begin{align*}
F = \Big\{\max_{j\in T}||\PPhi_{j}^H\n||_2>\sqrt{C_{\alpha}\sigma^2\mu_{max}^2K\left[\log(J)+\log(N)\right]}\Big\}
\end{align*}
where $C_{\alpha}=(4\sqrt{2}+4)\alpha+4$, we have $\text{Pr}(F\mid E)+\text{Pr}(E^C)\leq N^{-\alpha+1}$. Therefore,

\begin{equation*}
\begin{aligned}
&\quad\text{Pr}\left(||\PPhi_{T}^H\n||_{2,\infty}\leq\sqrt{C_{\alpha}\sigma^2\mu_{max}^2K\left[\log(J)+\log(N)\right]}\right)\\
&= \text{Pr}(F^C)\geq\text{Pr}(F^C\cap E)=1-[\text{Pr}(E^C)+\text{Pr}(F\cap E)]\\
&\geq 1-[\text{Pr}(E^C)+\text{Pr}(F| E)]\geq 1-N^{-\alpha+1}.
\end{aligned}
\end{equation*}
\end{proof}

\subsection{Proof of Theorem \ref{maintheorem}}
We now sum up the related lemmas to derive the final results in Theorem \ref{maintheorem}. By setting $\delta=\frac{1}{2}$, Lemma \ref{isometryproperty} shows that $\tPPhi_T^H\tPPhi_T$ is invertible and $||(\tPPhi_T^H\tPPhi_T)^{-1}||\leq2$ with probability at least $1-N^{-\alpha+1}$ when $N\geq C_{\alpha,0}\mu_{max}^2KJ\log^2(N)$ for $\alpha>1$.

By applying the same $\alpha$ to Lemma \ref{stc1} and \ref{stc2} and setting $\gamma=\frac{1}{2}$, we can get that, conditioned on $\tPPhi_T^H\tPPhi_T$ being invertible and $||(\tPPhi_T^H\tPPhi_T)^{-1}||\leq2$, $||\s_{T^C}||_{2,\infty}\leq\frac{1}{2}$ which implies that the support of the unique optimal solution $\hat{\X}$ to \eqref{glasso} is contained within the support of the ground truth solution $\X_0$, with probability at least $1-2N^{-\alpha+1}$, when
\begin{align}
\lambda\geq\sqrt{C_{\alpha,2}\sigma^2\mu_{max}^2K[\log(M-J)+\log(N)]}
\label{lambdastc}
\end{align}
and
\begin{align*}
N\geq C_{\alpha,3}\mu_{max}^2JK[\log(M-J)+\log(N)].
\end{align*}

As for the recovery error, we apply the same $\alpha$ to Lemma \ref{phinlemma} and substitute \eqref{phinbound} into \eqref{2infdeltax}. As a result, conditioned on $\tPPhi_T^H\tPPhi_T$ being invertible and the support of the unique optimal solution $\hat{\X}$ being contained within the support of $\X_0$,
\begin{equation}
\begin{aligned}
&\quad||\hat{\X}-\X_{0}||_{2,\infty}\leq4\sqrt{J}\left(||\PPhi_{T}^H\n||_{2,\infty}+\lambda\right)\\
&\leq \sqrt{C_{\alpha}\sigma^2\mu_{max}^2JK\left[\log(J)+\log(N)\right]}+4\sqrt{J}\lambda\\
\end{aligned}
\label{finaldeltax}
\end{equation}
where $C_{\alpha}$ is a constant that grows linearly with $\alpha$, with probability at least $1-N^{-\alpha+1}$ when $N\geq 10\log(J)+10\alpha\log(N)$.

Therefore, after combining the probability and the requirement on $N$ and $\lambda$, we can conclude that, with probability at least $1-4N^{-\alpha+1}$, \eqref{glasso} has a unique optimal solution $\hat{\X}$ with its support contained within the support of the ground truth solution $\X_0$ and the recovery error in terms of $\ell_{2,\infty}$-norm satisfies \eqref{finaldeltax} when $\lambda$ satisfies \eqref{lambdastc} and $N\geq C_{\alpha,1}\mu_{max}^2JK[\log(M-J)+\log^2(N)]$ where $C_{\alpha,1}=\max\{C_{\alpha,0},C_{\alpha,3}\}$ for $\alpha>1$.

\section{Numerical Simulations\label{simulations}}

In this section, we present several numerical simulations to demonstrate and support the theoretical results in Theorem \ref{maintheorem}. In these simulations, each entry of the dictionary $\A\in \R^{N\times M}$ is sampled independently from the standard normal distribution and $\B\in\C^{N\times K}$ contains the first $K$ columns of the normalized $N\times N$ DFT matrix. The real and imaginary components of $c_j\in \C$ and $\h_j\in\C^{K\times 1}$ follow the i.i.d standard normal distribution and the support, $T$ with $|T|=J$, of the ground truth solution $\X_0=[c_1\h_1, \cdots,\c_M\h_M]\in \C^{K\times M}$ is selected uniformly at random. Problem \eqref{glasso} is solved via CVX \cite{grant2008cvx}.

\subsection{Range of $\lambda$ for Exact Support Recovery}

In the first simulation, we determine the effective range of $\lambda$ for exact support recovery. Theoretically, \eqref{lambdabelow} provides a lower bound for $\lambda$ such that Theorem \ref{maintheorem} holds and \eqref{x0below} gives an upper bound to achieve exact support recovery. To verify the bounds of $\lambda$, we define $\gamma_0 = \sqrt{\sigma^2\mu_{max}^2K\left[\log(M-J)+\log(N)\right]}$ and $\gamma = \frac{\gamma_0}{\min_{j\in T}||\x_{0,j}||_2}$. \eqref{lambdabelow} implies that we could set $\lambda=k\gamma_0$ for some $k>0$. In addition, according to \eqref{lambdabelow} and \eqref{x0below} in Theorem \ref{maintheorem}, when all system parameters except $\lambda$ are fixed, to achieve exact support recovery, $\lambda$ should satisfy
\begin{align}
C_1\gamma_0\leq\lambda=k\gamma_0<\frac{\min_{j\in T}||\x_{0,j}||_2-C_2}{C_3}
\end{align}
which is equivalent to
\begin{align*}
C_1\leq k<\frac{C_4}{\gamma}-C_5
\end{align*}
where $C_4=\frac{1}{C_3}$ and $C_5 = \frac{C_2}{C_3\gamma_0}$. To examine this relation, we fix $\sigma=0.1$, $J=K=3$, $N=100$, and $M=150$, and we vary $k$ and $\gamma$. 50 trials are run for each $(k,\gamma)$ pair and we record the exact support recovery rate in Fig. \ref{fig:lambda}, from which we do observe that $k$ should be larger than a constant which is approximately $1.2$ under this setting and that $k$ has a reciprocal relation with $\gamma$.

\begin{figure}[h]
\begin{center}
   \includegraphics[width=1\linewidth]{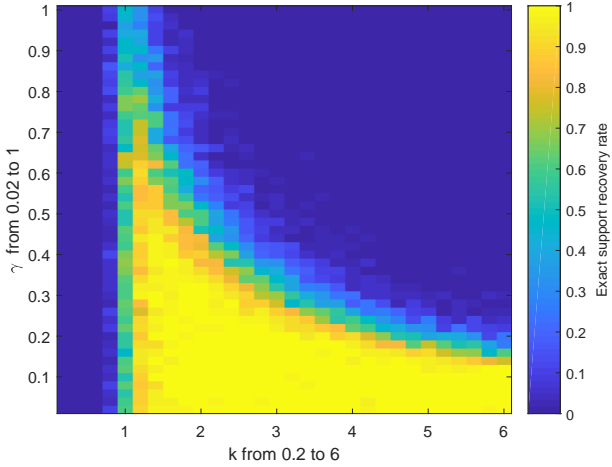}
\end{center}
   \caption{ The relation between $k$ and $\gamma$ in terms of the exact support recovery rate where $\lambda=k\gamma_0$ and $\gamma = \frac{\gamma_0}{\min_{j\in T}||\x_{0,j}||_2}$. }
\label{fig:lambda}
\end{figure}

\subsection{Number of Observations $N$ for Exact Support Recovery}

Equation \eqref{N} in Theorem \ref{maintheorem} indicates that the sufficient number of observations, $N$, scales nearly linearly with respect to the subspace dimension $K$ and the sparsity $J$. To verify that, in the second simulation, we set $M=150$, $k=3$, and $\gamma=0.02$ to make sure that $\lambda$ is in an appropriate range for exact support recovery. We vary $N$ and $K$ (with fixed $J=3$) and record the exact support recovery rate in Fig. \ref{fig:NK}. The result of a similar simulation but varying $N$ and $J$ (with fixed $K=3$) is shown in Fig. \ref{fig:NJ}. 50 simulations are run for each setting. We observe that linear scaling of the number of observations $N$ with the subspace dimension $K$ and the sparsity $J$ is sufficient for exact support recovery.

\begin{figure}[h]
\begin{center}
   \includegraphics[width=1\linewidth]{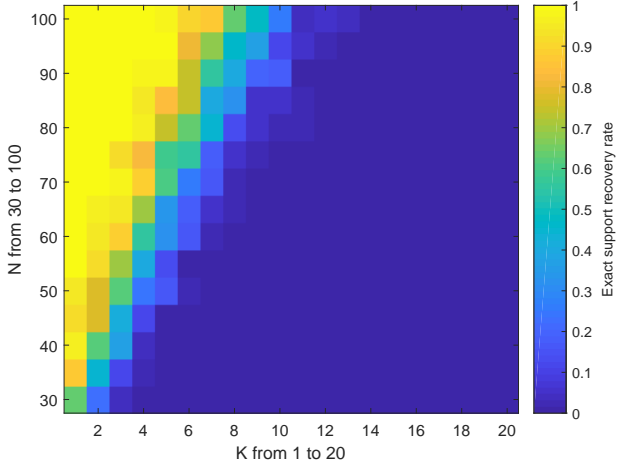}
\end{center}
   \caption{ The nearly linear relation between the number of observations, $N$, and the subspace dimension, $K$, to achieve exact support recovery. }
\label{fig:NK}
\end{figure}

\begin{figure}[h]
\begin{center}
   \includegraphics[width=1\linewidth]{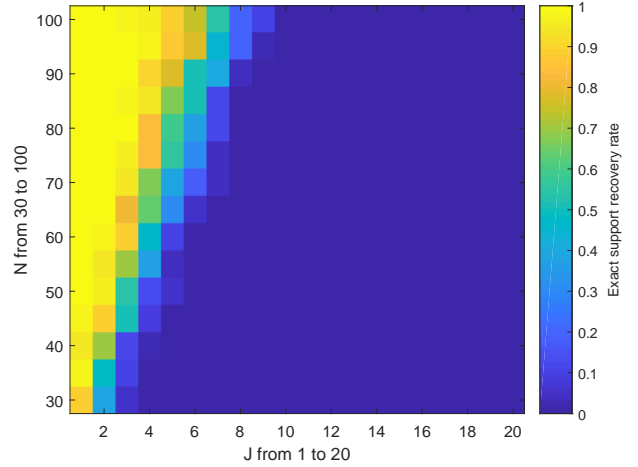}
\end{center}
   \caption{ The nearly linear relation between the number of observations, $N$, and the sparsity, $J$, to achieve exact support recovery. }
\label{fig:NJ}
\end{figure}

\subsection{Recovery Error Bound}
Next we turn to verify the recovery error bound in \eqref{errorbound}, which scales linearly with respect to $\lambda$ and nearly linearly with respect to $\sqrt{J}$. We set $K=3$, $N=100$, $M=150$, and $\gamma=0.02$. In Fig. \ref{fig:RElambda}, we use $\lambda=k\gamma_0$  (with fixed $J=3$) and vary $k$ within the proper range for exact support recovery based on Fig. \ref{fig:lambda}. 100 trials are run for each $k$ and we calculate the mean and standard deviation of the recovery error $||\hat{\X}-\X_0||_{2,\infty}$. Note that the recovery error is counted only when the exact support recovery is achieved. In this figure, we do observe linear scaling of the error with $\lambda$.

Similarly, we vary $J$ (with fixed $\lambda=3\gamma_0$) within the proper range for exact support recovery based on Fig. \ref{fig:NJ} and record the squared recovery error $||\hat{\X}-\X_0||_{2,\infty}^2$ in Fig. \ref{fig:REJ}. Again, the squared recovery error is counted only when the exact support recovery is achieved. In this figure, we can observe nearly linear scaling of the squared error with $J$.

\begin{figure}[h]
\begin{center}
   \includegraphics[width=1\linewidth]{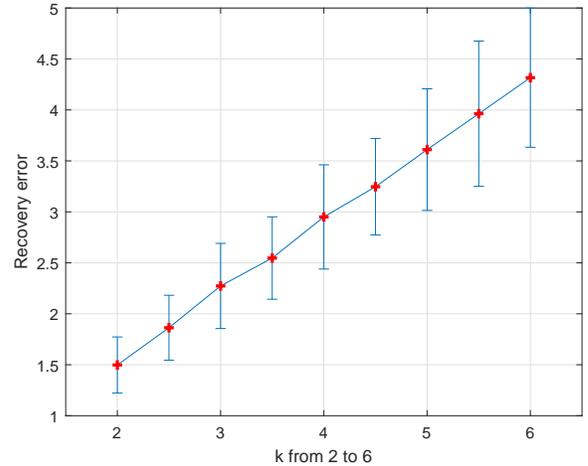}
\end{center}
   \caption{The linear relation between the recovery error, $||\hat{\X}-\X_0||_{2,\infty}$, and the regularization parameter $\lambda=k\gamma_0$. The red plus signs and the blue horizontal sticks indicates the mean and standard deviation of the recovery error.}
\label{fig:RElambda}
\end{figure}

\begin{figure}[h]
\begin{center}
   \includegraphics[width=1\linewidth]{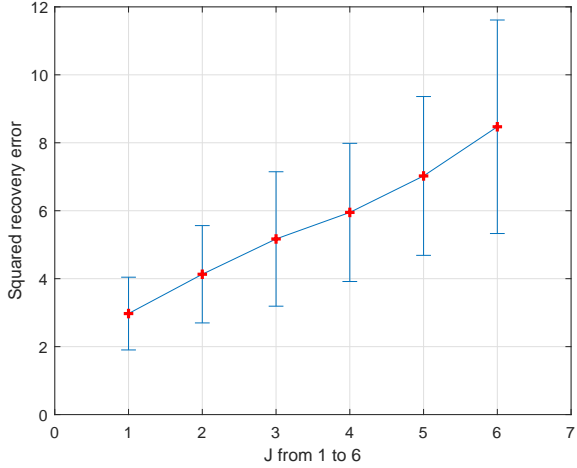}
\end{center}
   \caption{The nearly linear relation between the squared recovery error, $||\hat{\X}-\X_0||^2_{2,\infty}$, and the sparsity $J$. The red plus signs and the blue horizontal sticks indicates the mean and standard deviation of the squared recovery error.}
\label{fig:REJ}
\end{figure}

\subsection{Single Molecule Imaging\label{SMI}}

In this experiment, we apply our signal model \eqref{origin} to  the single molecule imaging problem and achieve super-resolution by solving \eqref{glasso}. In molecule imaging via stochastic optical reconstruction microscopy (STORM) \cite{rust2006sub}, the sub-cellular structures are dyed using fluorophores, and during each observation only a small portion of the fluorophores are activated. Moreover, fluorophores at different depths will undergo different degrees of blurring.

Consequently, each observed image frame consists of a few activated fluorophores convolved with the non-stationary Gaussian point spread functions of the microscope as shown in Fig. \ref{fig:PSF} (a). Specifically, the observed low resolution frame is of size $64\times64$ pixels and each pixel corresponds to a region of size 100$\times$100 nm. The goal is to construct a target image with $320\times320$ pixels with each pixel corresponding to a region of size $20\times20$ nm.

If we vectorize the frames, each observed low resolution frame can be represented as
\begin{align}
\y=Sample\left[\sum_{j=1}^M c_j(\B'\h_j)\circledast\e_j +\n'\right] \in\R^{N\times 1}
\label{molecule}
\end{align}
where $Sample[\cdot]$ indicates the sub-sampling operator, $N=64\times64=4096$, and $M=320\times320=102400$. Moreover, $c_j$ is the unknown fluorophore intensity at the $j$-th position, $\B'$ models the subspace containing the non-stationary Gaussian point spread functions (with unknown coefficient vector $\h_j$ for the $j$-th position), $\e_j\in\R^{M}$ is the $j$-th column of the identity matrix, and $\n'$ is the unknown additive noise. All observed frames, $\y$, come from the Single-Molecule Localization Microscopy grand challenge organized by ISBI \cite{ISBI2013}. The dataset contains 12000 low resolution frames, and the maximum number of activated fluorophores in each frame is 18 which implies that at most $J=18$ coefficients $c_j$ are non-zero for each $\y$.


To apply our model, we must construct the subspace, $\B'$, to capture the non-stationary point spread functions. By changing the variances (widths), we generate nine different Gaussian point spread functions; four examples are shown in Fig. \ref{fig:PSF} (b). We then apply the singular value decomposition (SVD) to a matrix of the vectorized point spread functions and record their singular values in Fig. \ref{fig:PSF} (c). From this we see that the point spread functions approximately live in a 3 dimensional subspace. Therefore, we set $K=3$ and let $\B'$ contain the singular vectors corresponding to the 3 largest singular values. We display the corresponding singular vectors in Fig. \ref{fig:PSF} (d).

\begin{figure}[h]
   \centering
   \subfloat[][An observed frame.]{\includegraphics[width=.2\textwidth]{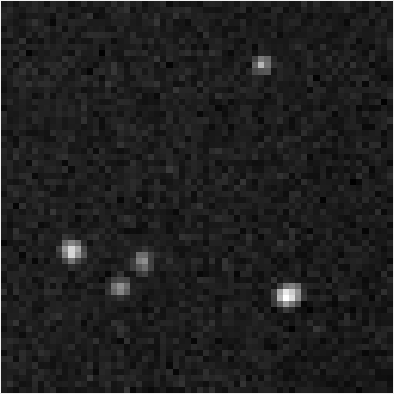}}
      \qquad
   \subfloat[][Point spread functions.]{\includegraphics[width=.2\textwidth]{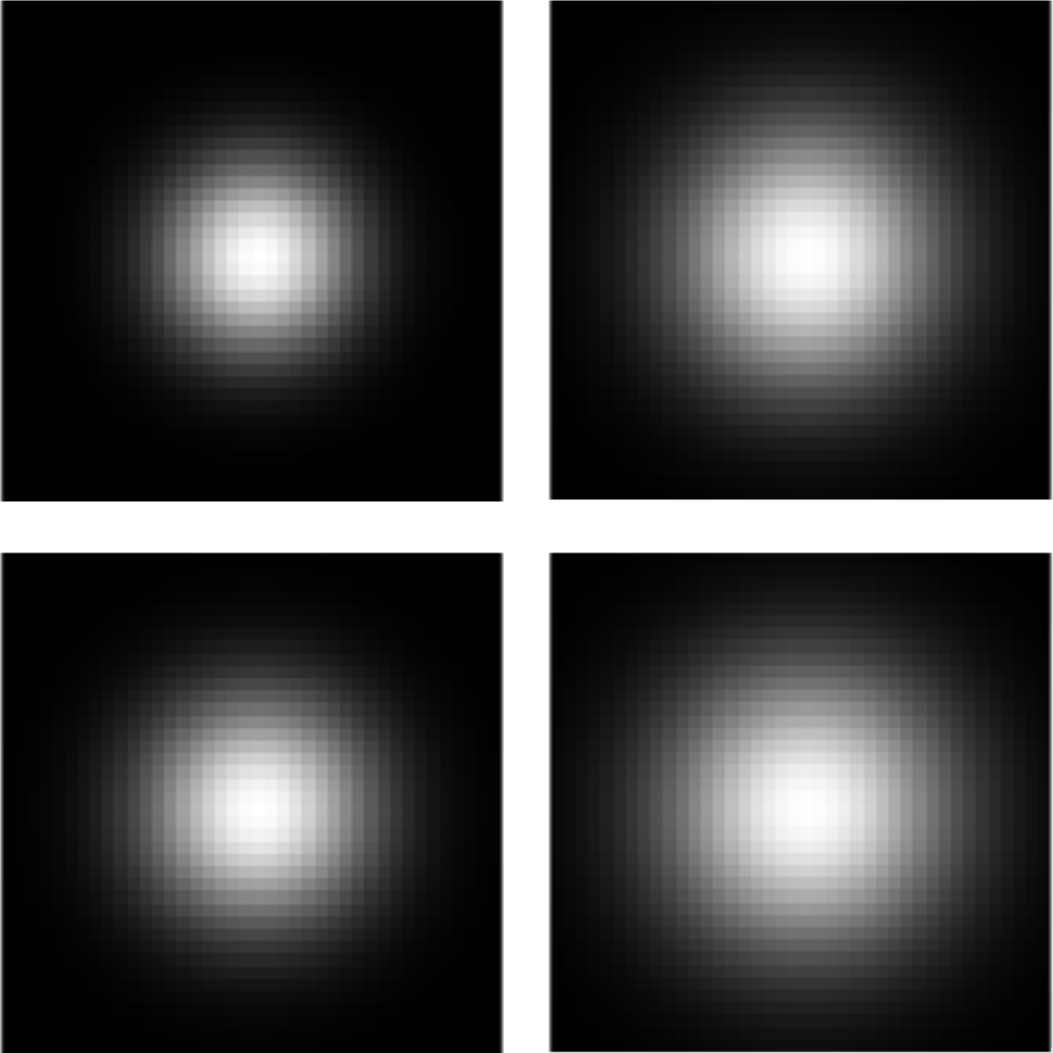}}
\\
      \subfloat[][The singular values.]{\includegraphics[width=.2\textwidth]{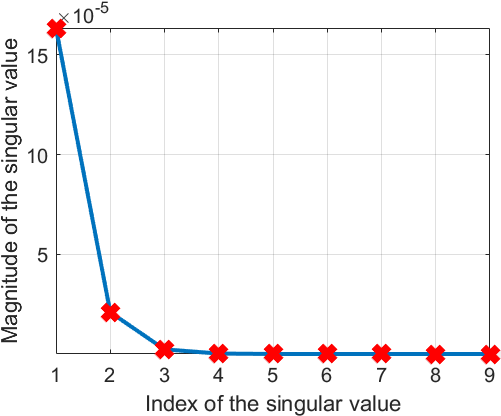}}
   \qquad
   \subfloat[][The singular vectors.]{\includegraphics[width=.2\textwidth]{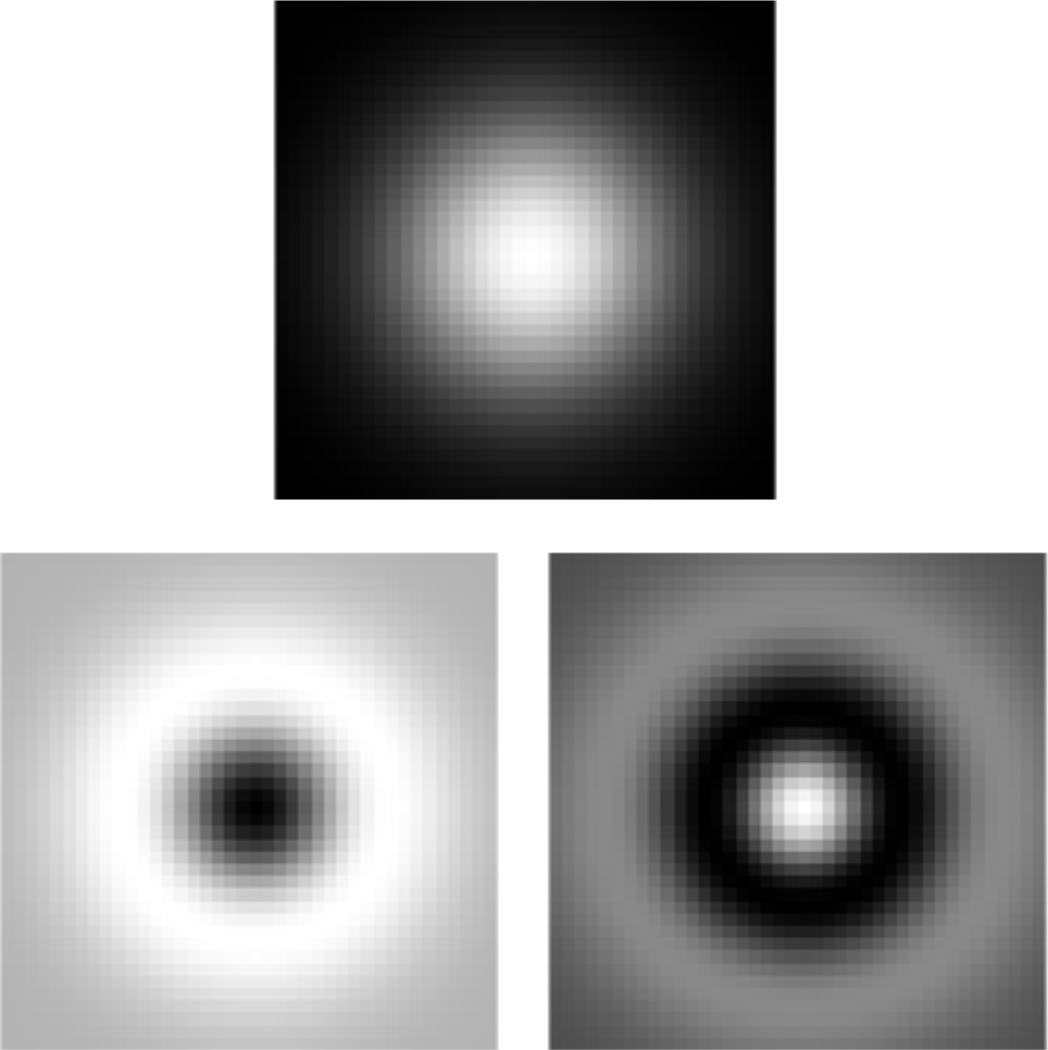}}
   \caption{The analysis of point spread functions. (a) A typical observed frame is of size $64\times64$ pixels with each pixel corresponding to a region of size $100\times100$ nm. (b) Four examples of the non-stationary point spread functions. (c) The singular values of the point spread functions. (d) The singular vectors corresponding to the three largest singular values.}
   \label{fig:PSF}
\end{figure}

To better illustrate the connection between the single molecule imaging problem and the signal model we study, \eqref{molecule} can be equivalently represented as
\begin{align*}
\y=Sample\Bigg\{IDFT\left[\sum_{j=1}^M c_j\D_j\a_j +\n\right]\Bigg\} \in\R^{N\times 1}
\end{align*}
where $IDFT[\cdot]$ denotes the inverse discrete Fourier transform operator, $\D_j=\diag(\B\h_j)$ where $\B=DFT[\B']$, $\a_j$s are the DFTs of spikes at all possible spatial locations, and $\n=DFT[\n']$. In this case, if we represent $\y=\L(\X)$ with $\X =[c_1\h_1,\cdots,c_M\h_M]$, the linear operator $\L$ incorporates additional inverse Fourier transform and sub-sample operators, and $\A$ is a Fourier dictionary instead of random Gaussian. The noise $\n$, $\h_j$, and $c_j$ for all $j$ are unknown, and the indices of the non-zero columns in $\X$ indicate the locations of the activated fluorophores in the high resolution image.


We pre-process each low resolution frame by subtracting the average intensity of the data set, and superimposing all the frame results in the low resolution image in Fig. \ref{fig:superresolution} (a). Moreover, we solve \eqref{glasso} for each observed low resolution frame via SpaRSA \cite{wright2009sparse}. By superimposing all the high resolution images that we get, we obtain the super-resolution result in Fig. \ref{fig:superresolution} (b). Although the dictionary is not Gaussian in this application, the superior super-resolution result verifies the effectiveness of the proposed signal model and minimization problem.

\begin{figure}[h]
   \centering
      \subfloat[][Low resolution input.]{\includegraphics[width=.2\textwidth]{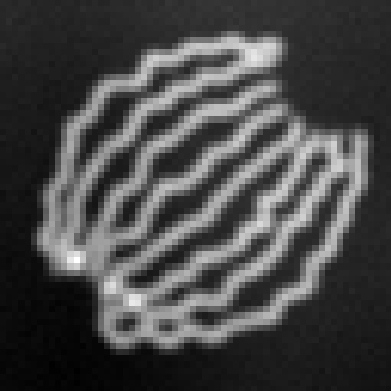}}
   \qquad
   \subfloat[][Super-resolution result.]{\includegraphics[width=.2\textwidth]{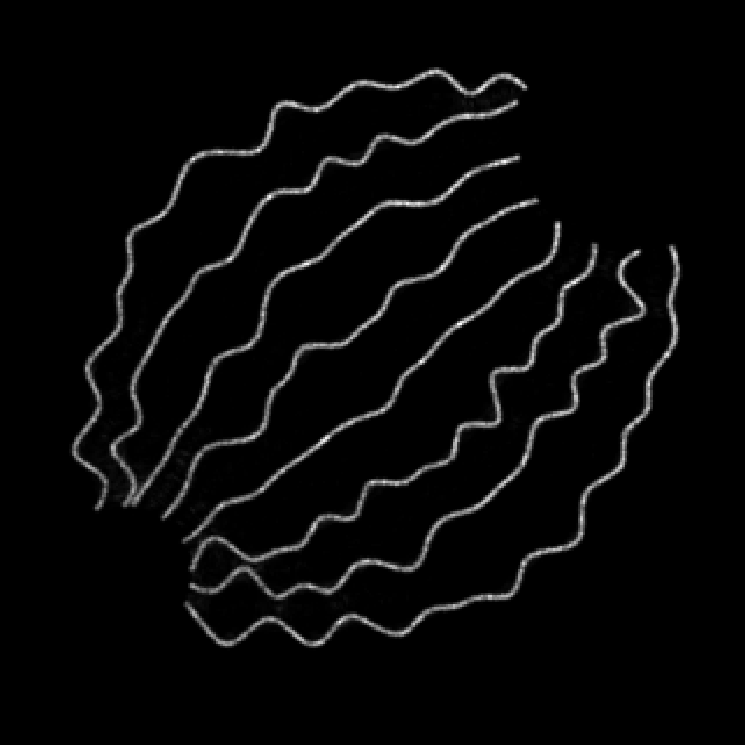}}
   \caption{The single molecule imaging experiment. The image in (a) is of size $64\times64$ pixels with each pixel corresponding to a region of size $100\times100$ nm. (b) shows the super-resolution result, which has size $320\times320$ pixels with each pixel corresponding to a region of size $20\times20$ nm.}
   \label{fig:superresolution}
\end{figure}

Finally, when $K=1$, \eqref{glasso} degenerates to the classical $\ell_1$-norm constrained lasso problem which has been comprehensively studied. However, by choosing $K=1$, the model sacrifices its ability to capture non-stationary modulation, which is significant in this problem when the point spread functions have several comparable singular values. Although in our case, we happen to have one dominant singular value as shown in Fig. \ref{fig:PSF} (c), which implies that super-resolution can be attempted with $K=1$, we see that a larger $K$ still benefits the super-resolution process. To demonstrate this, we run the single molecule imaging experiments again using $K=3$ and $K=1$. Three super-resolution examples are shown in Fig. \ref{fig:compare}, from which we can find that although $K=3$ and $K=1$ achieve similar performance, some activated fluorophores can be more accurately represented using the 3-dimensional subspace ($K=3$), and that leads to a more clear and accurate super-resolution result.


\begin{figure}[h]
   \centering
   \subfloat[][Input frame.]{\includegraphics[width=.1\textwidth]{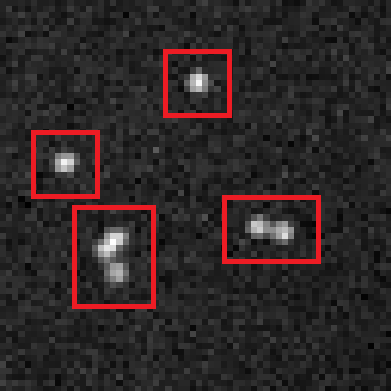}}
      \quad
   \subfloat[][Result for $K=3$.]{\includegraphics[width=.17\textwidth]{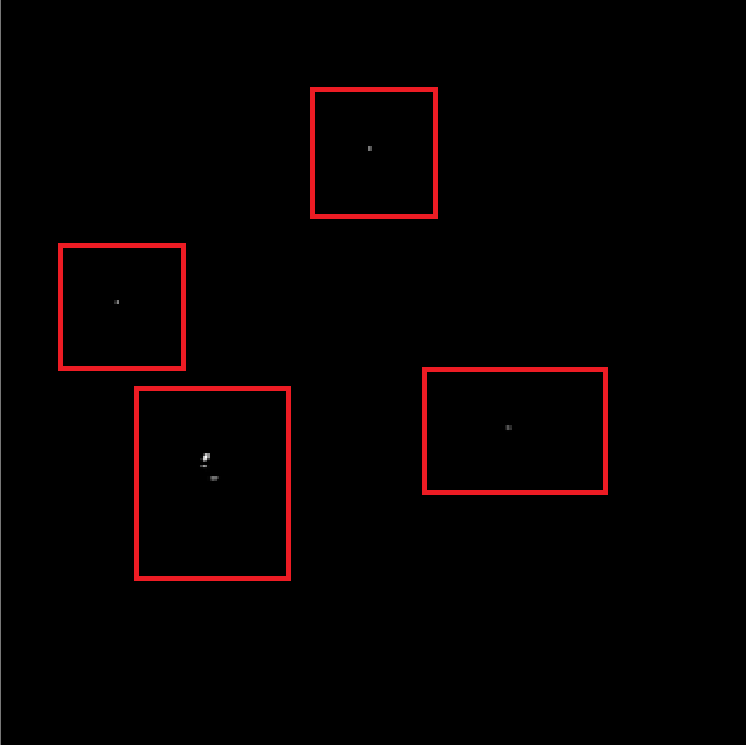}}
   \quad
      \subfloat[][Result for $K=1$.]{\includegraphics[width=.17\textwidth]{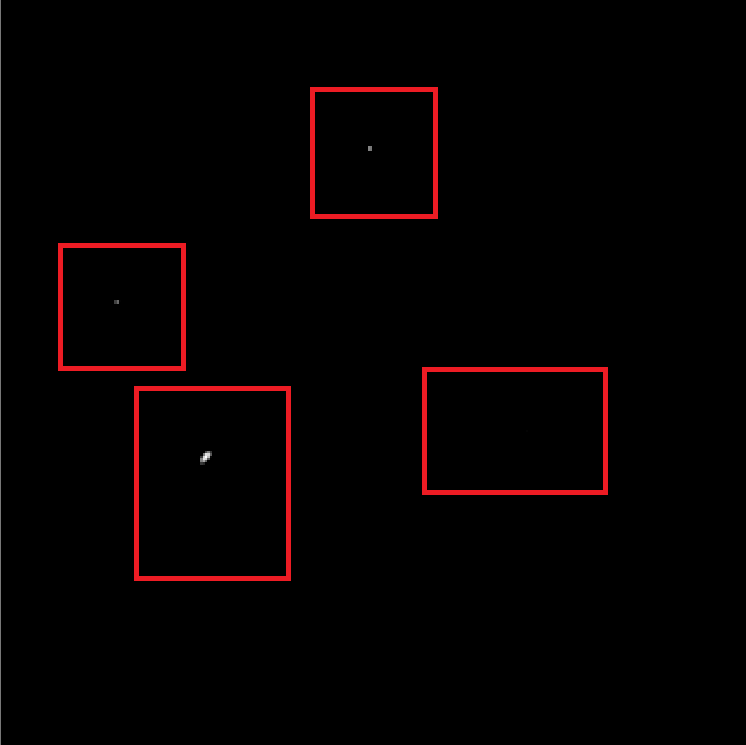}}
\\
   \subfloat[][Input frame.]{\includegraphics[width=.1\textwidth]{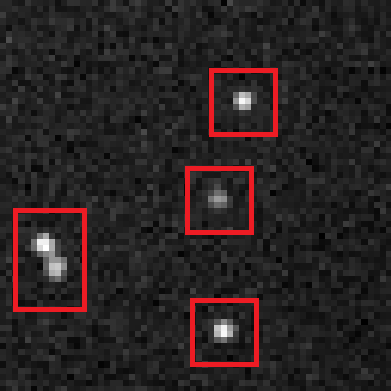}}
      \quad
   \subfloat[][Result for $K=3$.]{\includegraphics[width=.17\textwidth]{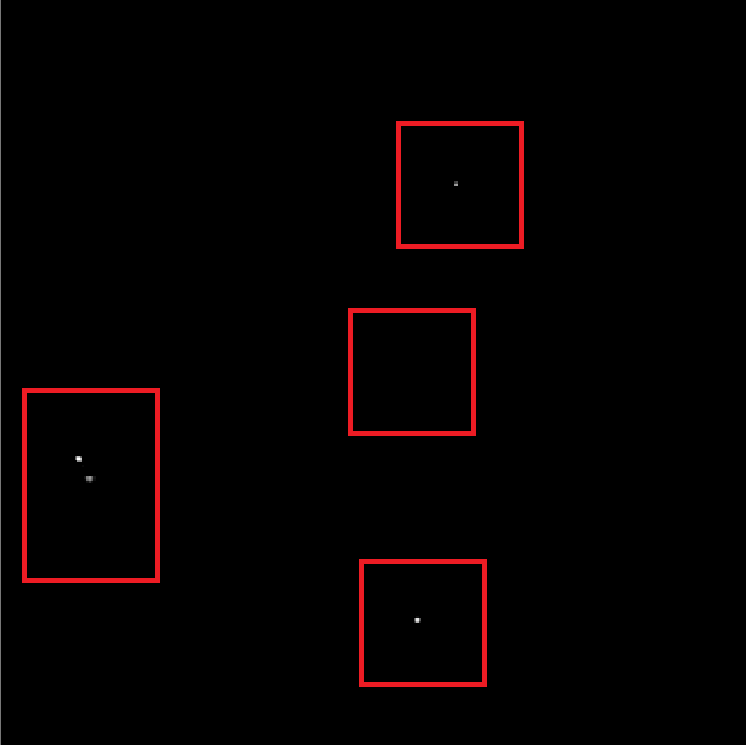}}
   \quad
      \subfloat[][Result for $K=1$.]{\includegraphics[width=.17\textwidth]{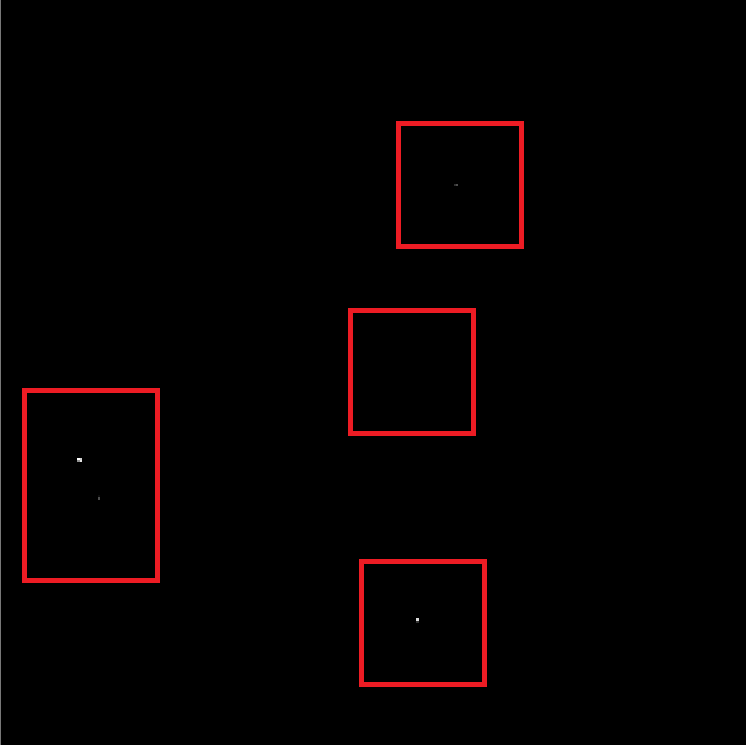}}
\\
   \subfloat[][Input frame.]{\includegraphics[width=.1\textwidth]{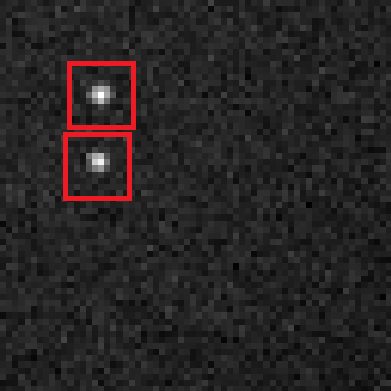}}
      \quad
   \subfloat[][Result for $K=3$.]{\includegraphics[width=.17\textwidth]{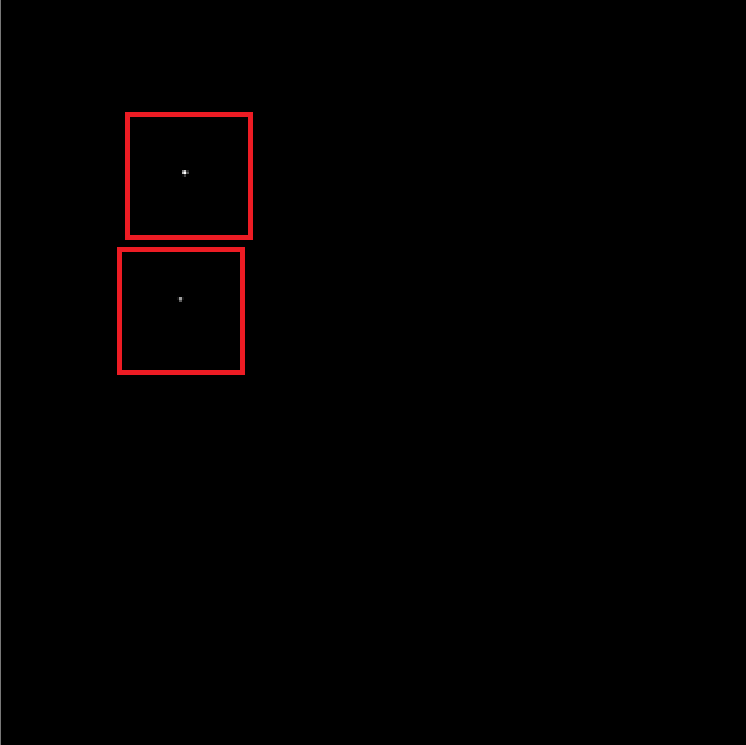}}
   \quad
      \subfloat[][Result for $K=1$.]{\includegraphics[width=.17\textwidth]{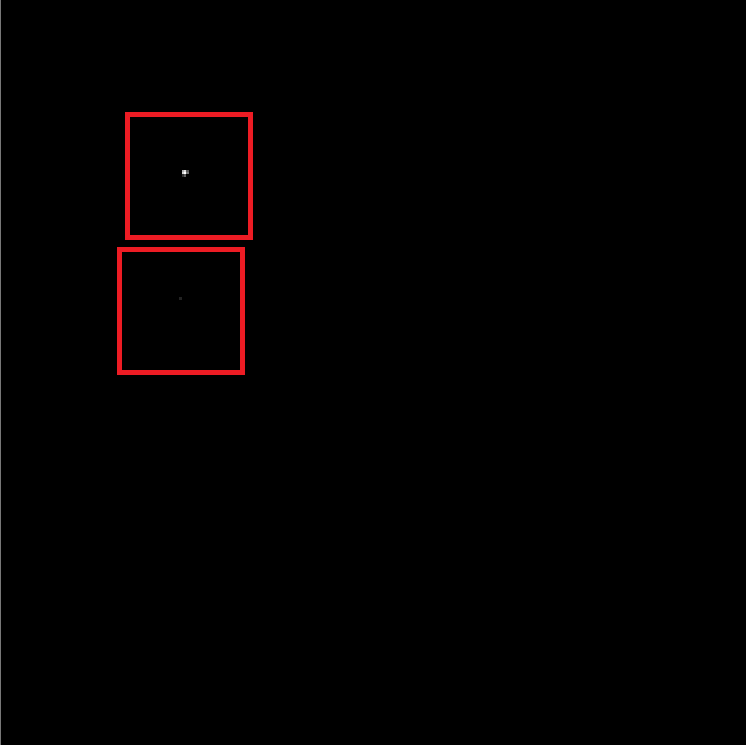}}
   \caption{Comparison between the super-resolution results for $K=3$ and $K=1$. (a), (d), and (g) are three low-resolution input frames. (b), (e), and (h) show the super-resolution results for $K=3$. (c), (f), and (i) show the super-resolution results for $K=1$. The area of interest is highlighted using the red rectangle. The input frames are of size $64\times64$ pixels and the outputs are $320\times320$ pixels.}
   \label{fig:compare}
\end{figure}









\section{Conclusion\label{conclusion}}

In this paper, we consider the problem of recovering a sparse signal with unbounded noise and non-stationary blind modulation. Using the lifting technique and with a subspace assumption on the modulating signals, we recast this problem as the recovery of a column-wise sparse matrix from structured linear observations. We apply $\ell_{2,1}$-norm regularized quadratic minimization, also known as the group lasso, to solve this problem and derive sufficient conditions on the sample complexity and regularization parameter for exact support recovery. We also bound the recovery error in terms of the $\ell_{2,\infty}$-norm. Numerical simulations are consistent with our predictions and support the theoretical results. Moreover, we apply our model to single molecule imaging and achieve promising super-resolution results. One useful generalization of the results in this paper would be to consider the random Fourier dictionary. Allowing $\D_j$ to be non-diagonal but live in a low-dimensional matrix subspace is another important generalization, which could open up other potential applications.

\section*{Acknowledgment}

This work was supported by NSF grant CCF$-1704204$. 



%
\bibliographystyle{ieeetr}
\bibliography{IEEEabrv,test}

%




\end{document}